\documentclass[11pt,reqno]{article}

\usepackage{amsmath, amsthm, amsopn, amssymb}
\usepackage{enumerate, enumitem, tikz, etoolbox, intcalc, geometry, caption, subcaption, afterpage}
\usepackage[english]{babel}
\usepackage{mathrsfs, mathtools, mathabx}
\usepackage{dsfont}
\usepackage{upgreek}

\DeclarePairedDelimiter\abs{\lvert}{\rvert}%
\DeclarePairedDelimiter\norm{\lVert}{\rVert}%
\DeclareMathOperator{\dist}{dist}
\DeclareMathOperator{\supp}{supp}

\DeclareMathOperator{\per}{per}
\DeclareMathOperator{\p}{p}
\DeclareMathOperator{\free}{free}
\DeclareMathOperator{\f}{f}
\DeclareMathOperator{\fp}{fp}
\DeclareMathOperator{\UP}{SG}
\DeclareMathOperator{\zero}{zero}

\DeclareMathOperator{\GFF}{GFF}
\DeclareMathOperator{\IV}{IV}
\DeclareMathOperator{\Vil}{Vil}
\DeclareMathOperator{\ext}{ext}

\DeclareMathOperator{\sep}{sep}
\DeclareMathOperator{\zv}{z}

\newtheorem{thm}{Theorem}[section]
\newtheorem{lemma}[thm]{Lemma}
\newtheorem{prop}[thm]{Proposition}
\newtheorem{definition}[thm]{Definition}

\newtheorem{claim}[thm]{Claim}

\theoremstyle{definition}
\newtheorem*{remark}{Remark}

\newtheorem*{conj}{Conjecture}

\newcommand{\twopartdef}[4]
{
	\left\{
		\begin{array}{ll}
			#1 & #2 \\
			#3 & #4
		\end{array}
	\right.
}
\newcommand{\threepartdef}[6]
{
	\left\{
		\begin{array}{lll}
			#1 & #2 \\
			#3 & #4 \\
			#5 & #6
		\end{array}
	\right.
}

\numberwithin{equation}{section}

\title{The Fr\"ohlich-Spencer Proof of the Berezinskii-Kosterlitz-Thouless Transition}
\author{Vital Kharash\thanks{Supported by ISF grant~861/15 and by ERC starting grant 678520 (LocalOrder). School of Mathematical Sciences, Tel Aviv
University, Tel Aviv 69978, Israel. Emails: vitalk@gmail.com,
peledron@post.tau.ac.il} \and Ron Peled\footnotemark[1] }

\begin{document}
\newcommand*{\mida}{d\mu_{\text{\tiny$\beta$}}(\phi)}
\newcommand*{\midaLambda}{d\mu_{\text{\tiny$\beta, \Lambda$}}^{\GFF}(\phi)}
\newcommand*{\midaLambdaV}{d\mu_{\text{\tiny$\beta, \Lambda, v$}}^{\GFF}(\phi)}
\newcommand*{\midaLambdaVof}[1]{d\mu_{\text{\tiny$\beta, \Lambda, v$}}^{\GFF}(#1)}
\newcommand*{\midaLambdaVBeta}[1]{d\mu_{\text{\tiny$#1, \Lambda, v$}}^{\GFF}(\phi)}
\newcommand*{\midar}{d\mu_\beta(\phi)}
\newcommand*{\ens}{\mathscr{N}}
\newcommand*{\ensE}{\mathscr{E}}
\newcommand*{\ensQ}{\mathcal{Q}}
\newcommand*{\ensN}{\mathscr{G}}
\newcommand*{\ensI}{\mathscr{I}}
\newcommand*{\ensS}{\mathscr{S}}

\newcommand*{\zed}{z(\beta, \varrho, \ens)}
\newcommand*{\tp}{\text{\tiny$\prime$}}
\newcommand*{\phiOf}[1]{\langle\phi, #1\rangle}
\newcommand*{\psiOf}[1]{\langle\psi, #1\rangle}
\newcommand*{\mOf}[1]{\langle m, #1\rangle}
\newcommand*{\sigmaOf}[1]{\langle\sigma, #1\rangle}
\newcommand*{\sigmaOfVil}[1]{\langle\sigma^{\Vil}, #1\rangle}

\newcommand*{\SqrsSet}{\mathscr{S}_k(\varrho)}
\newcommand*{\SqrsSetTag}{\mathscr{S}^{\sep}_k(\varrho)}
\newcommand*{\SqrsSetTagGen}{\mathscr{S}^{\sep}}
\newcommand*{\SqrsSetDtag}{\mathscr{S}^{\prime\prime}_k(\varrho)}
\newcommand*{\SqrsSetGamma}{\mathscr{S}_{\gamma(k)}(\varrho)}
\newcommand*{\SqrsSetTagGamma}{\mathscr{S}^{\sep}_{\gamma(k)}(\varrho)}
\newcommand*{\SqrsSetDtagGamma}{\mathscr{S}^{\prime\prime}_{\gamma(k)}(\varrho)}
\newcommand*{\SqrsSetOf}[2]{\mathscr{S}_{#1}(#2)}
\newcommand*{\SqrsSetTagOf}[2]{\mathscr{S}^{\sep}_{#1}(#2)}
\newcommand*{\gradNorm}[1]{\sum_{j\thicksim \ell} ({#1}_j - {#1}_\ell)^2}
\newcommand*{\gradNormUp}[1]{\sum_{j\thicksim \ell} ({#1}(j) - {#1}(\ell))^2}
\newcommand*{\gradNormOld}[1]{{\langle #1, -\Delta_\Lambda #1\rangle}}

\newcommand*{\bvr}{\bar{\varrho}}

\newcommand*{\SG}{(\Gamma, \eta, \theta)\text{-sub-Gaussian}}
\newcommand*{\SGS}{(\Gamma, \eta, \theta)\text{-sub-Gaussians}}

\newcommand*{\DOne}{D_1}
\newcommand*{\DTwo}{D_2}
\newcommand*{\DThr}{D_3}
\newcommand*{\DFr}{D_4}
\newcommand*{\DSix}{D_5}
\newcommand*{\DSev}{D_6}

\newenvironment{enumerate-math}
{\begin{enumerate}
\addtolength{\itemsep}{5pt}
\renewcommand\theenumi{(\roman{enumi})}}
{\end{enumerate}}

\maketitle
\begin{abstract}
  We present the Fr\"ohlich-Spencer proof of the Berezinskii-Kosterlitz-Thouless transition. Our treatment includes the proof of delocalization for the integer-valued discrete Gaussian free field at high temperature and the proof of existence of a phase with power-law decay of correlations in the plane rotator model with Villain interaction, both in two dimensions.
  The treatment differs from the original in various technical points and we hope it will be of benefit to the community.
\end{abstract}

\section{Introduction}\label{sec:introduction}
Berezinskii \cite{Ber71,Ber72} and Kosterlitz and Thouless \cite{KT72,KT73} predicted in 1972 the existence of new types of phase transitions leading to \emph{topological phases of matter}. These theoretical discoveries were the basis for the award of the 2016 Nobel prize in Physics to Haldane, Kosterlitz and Thouless. Fr\"ohlich and Spencer \cite{fs} proved these predictions with mathematical rigor in a celebrated 1981 paper. The Fr\"ohlich-Spencer proof yielded a wealth of information on many two-dimensional models of statistical physics: Delocalization for the integer-valued discrete Gaussian free field and the Solid-On-Solid model at high temperature, the absence of Debye screening in the Coulomb gas at low temperature, the existence of a phase with power-law decay of correlations (the Berezinskii-Kosterlitz-Thouless phase) in the plane rotator models with XY and Villain interactions at low temperature and in the $\mathbb{Z}_n$ clock models at large $n$ and intermediate temperatures. Nevertheless, mathematical understanding of the fine properties of these models remains incomplete with many outstanding open questions (see Section~\ref{sec:openQ}). In this paper we present the Fr\"ohlich-Spencer proof of the delocalization of the integer-valued discrete Gaussian free field and the existence of the Berezinskii-Kosterlitz-Thouless phase in the plane rotator model with Villain interaction. While our proof is essentially the same as the original, we have attempted to make the presentation accessible by providing a detailed overview of the proof, organizing the arguments into additional intermediate steps and making various local enhancements. We additionally chose to provide all arguments in finite volume (with estimates uniform in the volume) to reduce the technical prerequisites on the side of the reader. A review of some of the main differences from the original paper is in Section~\ref{sec:diffs_from_fs}. It is hoped that this paper will be of benefit to the community and provide further access to the remaining challenges.

\subsection{Notation}
In presenting the main results and throughout the paper we use the following notation.

Let $\Lambda_L^{\free}$, $L>1$, be the graph whose vertices are $V(\Lambda_L^{\free}) = \{0, \ldots, L-1 \}^2$, and whose edges $E(\Lambda_L^{\free})$ are all pairs $\{(a,b),(c,d)\}$ with $\{|a-b|, |c-d|\}=\{0,1\}$.

Let $\Lambda_L^{\per}$, $L > 1$ even, be the graph whose vertices are $V(\Lambda_L^{\per}) = \{0, \ldots, L-1 \}^2$, and whose edges $E(\Lambda_L^{\per})$ are all pairs $\{(a,b), (c,d)\}$ with $(a,b),(c,d)$ equal in one coordinate and differing by exactly one modulo $L$ in the other coordinate. We require $L$ to be even so that $\Lambda_L^{\per}$ is bipartite.

Let $\Lambda_L^{\zero}$, $L>1$, be the multigraph whose vertices are $V(\Lambda_L^{\zero}) = V(\Lambda_L^{\free}) \bigcup \{\zv\}$ and whose edges are $E(\Lambda_L^{\free})$ adjoined with an edge $\{j,\zv\}$, for each $j\in V(\Lambda_L^{\free})$, of multiplicity $4$ minus the degree of $j$ in $\Lambda_L^{\free}$ (so that each $j\in V(\Lambda_L^{\zero})$ has degree $4$ in $\Lambda_L^{\zero}$).

We call $\Lambda$ a \textit{square domain} if $\Lambda$ is either $\Lambda_L^{\free}$, $\Lambda_L^{\per}$ or $\Lambda_L^{\zero}$ for some $L>1$ ($L$ even for $\Lambda_L^{\per}$) and distinguish the cases by saying that $\Lambda$ has free, periodic or zero boundary conditions (b.c.), respectively.

For a finite, connected graph $\Lambda$ and $j,\ell\in V(\Lambda)$ we write $j\thicksim \ell$ to denote that $j$ and $\ell$ are adjacent. We write $\dist(j,\ell)$ for the graph distance between $j$ and $\ell$ and $\dist(A,B):=\min_{j\in A, \ell\in B}\dist(j,\ell)$ for $A,B\subset V(\Lambda)$. We define the finite difference Laplacian $\Delta_\Lambda$ of $\Lambda$ as the linear operator satisfying
\begin{equation}\label{eq:Laplacian_def}
(\Delta_\Lambda f)(j) := \sum_{\ell:\ell \thicksim j} f_\ell-f_j
\end{equation}
(with both terms inside the sum, so that $f(j)$ is multiplied by the degree of $j$). We define $\Delta_\Lambda^{-1}$ to be the Moore-Penrose pseudoinverse of $\Delta_\Lambda$, characterized by requiring that $\Delta_\Lambda^{-1}$ be linear and
\begin{equation}\label{eq:pseudo-inverse_def}
\text{$\Delta_\Lambda \Delta_\Lambda^{-1}f = \Delta_\Lambda^{-1} \Delta_\Lambda f = f$ for every $f:\Lambda\rightarrow\mathbb{R}$ with $\sum_{j\in\Lambda}f(j)=0$,\quad $\Delta_\Lambda^{-1} f = 0$ for constant $f$.}
\end{equation}
We call $-\Delta_\Lambda^{-1}$ the Green function of $\Lambda$. For notational simplicity, we identify $\Lambda$ with $V(\Lambda)$ when no ambiguity can arise. For $f,g: \Lambda \rightarrow \mathbb{R}$ we write
$$
\langle f, g \rangle := \sum_{j\in\Lambda} f_j g_j.
$$

\subsection{Discrete Gaussian Free Field} \label{sec:DGFF}
The \emph{discrete (or lattice) Gaussian free field} on a finite, connected graph $\Lambda$ at inverse temperature $\beta>0$, normalized to take values in $[-\pi, \pi)$ at a given $v\in \Lambda$, is the measure on $\phi:\Lambda\rightarrow\mathbb{R}$ given by
\begin{equation} \label{label:mida}
\midaLambdaV := \frac{1}{Z_{\beta, \Lambda,v}^{\GFF}} \cdot
\exp\Bigg[ -\frac{\beta}{2}\sum_{j\thicksim\ell}(\phi_j - \phi_\ell)^2 \Bigg]
\mathds{1}_{[-\pi, \pi)}\big(\phi_v\big)\prod_{j \in \Lambda}d\phi_j,
\end{equation}
where $d\phi_j$ is the Lebesgue measure on $\mathbb{R}$, and the normalization constant $Z_{\beta, \Lambda,v}^{\GFF}$ makes $\mu_{\text{\tiny$\beta, \Lambda, v$}}^{\GFF}$ into a probability measure. We emphasize that in the sum over $j\thicksim\ell$ each edge of $\Lambda$ is counted exactly once. We write $\mathbb{E}_{\beta, \Lambda, v}^{\GFF}$ for the corresponding expectation.

The normalization $\phi_v \in [-\pi, \pi)$, instead of the more standard requirement that $\phi_v=0$, is chosen with a view towards Section~\ref{sec:DGFF_periodic_weights} below, where we augment the above definition with the addition of $2\pi$-periodic single-site weights. This choice does not affect the distribution of the gradients $(\phi_j - \phi_\ell)$, $j,\ell\in\Lambda$, for the discrete Gaussian free field.


For every $f:\Lambda \rightarrow \mathbb{R}$ satisfying $\sum_{j \in \Lambda} f_j = 0$ it holds that
\begin{equation} \label{eq:DGFFRes}
\mathbb{E}_{\beta, \Lambda, v}^{\GFF}\left[ e^{\phiOf{f}} \right] =
\exp\left({\frac{1}{2\beta} \langle f, -\Delta_\Lambda^{-1}f\rangle}\right),
\end{equation}
as one sees with the change of variables
$
\phi_j \rightarrow \phi_j + \sigma_j,
$
where $\sigma:\Lambda \rightarrow \mathbb{R}$ is the solution of
\begin{equation} \label{eq:sigmaDef}
\twopartdef
{(-\Delta_\Lambda) \sigma = \frac{1}{\beta} f} {}
{\sigma_v = 0}{}\!\!\!\!\!\!\!\!.
\end{equation}
The equality \eqref{eq:DGFFRes} remains valid also when $f$ is allowed to take complex values (but still sums to $0$), by analytic continuation.

\subsection{Integer-Valued Discrete Gaussian Free Field}
The \emph{integer-valued discrete Gaussian free field} on a finite, connected graph $\Lambda$ at inverse temperature $\beta>0$, normalized to equal $0$ at a given $v\in\Lambda$, is the measure on $m:\Lambda\rightarrow\mathbb{Z}$ given by
$$
d\mu_{\beta, \Lambda, v}^{\IV}(m) :=
\frac{1}{Z_{\beta, \Lambda,v}^{\IV}}
\exp\Bigg[-\frac{\beta}{2}\sum_{j\thicksim\ell}(m_j - m_\ell)^2\Bigg]
d\delta_0(m_v)
\prod_{j \in \Lambda\setminus\{v\}} d_{\text{count}(\mathbb{Z})}(m_j)
$$
where $\delta_0$ is the Dirac delta measure at $0$, $d_{\text{count}(\mathbb{Z})}$ is the counting measure on $\mathbb{Z}$, and the normalization constant $Z_{\beta, \Lambda,v}^{\IV}$ normalizes the measure $\mu_{\beta, \Lambda, v}^{\IV}$ to be a probability measure. We write $\mathbb{E}_{\beta, \Lambda, v}^{\IV}$ for the corresponding expectation.

Our main interest is in computing $\mathbb{E}_{\beta, \Lambda, v}^{\IV}\left[ F(m) \right]$ for functionals $F:\mathbb{Z}^{\Lambda}\rightarrow \mathbb{R}$ which are invariant under translation by constant functions (i.e., $F(m) = F(m + c)$ for constant $c:\Lambda\rightarrow\mathbb{Z}$). We note that for such functionals, the expectation is independent of the choice of $v\in\Lambda$, i.e.,
$$
\mathbb{E}_{\beta, \Lambda, v_1}^{\IV}\left[ F(m) \right] =
\mathbb{E}_{\beta, \Lambda, v_2}^{\IV}\left[ F(m) \right], \quad v_1,v_2\in\Lambda.
$$
It is well known that the integer-valued discrete Gaussian free field is localized at low temperature, as follows from standard Peierls-type arguments. For instance, that there exist absolute constants $\beta_0, C>0$ for which
\begin{equation*}
\mathbb{E}_{\beta, \Lambda_L^{\free}, (0,0)}^{\IV}\left[ (m_j-m_k)^2 \right]\le C,\quad\text{for all $L>1$, $j,k\in \Lambda_L^{\free}$ and $\beta\ge \beta_0$}.
\end{equation*}
Fr\"ohlich and Spencer \cite{fs} prove the following theorem which establishes the existence of a roughening transition for the integer-valued discrete Gaussian free field.
\begin{thm} \label{thm:IV}
For every $\varepsilon > 0$, there exists $\beta_0 > 0$, such that the following holds for any square domain $\Lambda$, with free or periodic b.c., and any $\beta < \beta_0$. For every $f:\Lambda \rightarrow \mathbb{R}$ with $\sum_{j \in \Lambda} f_j = 0$,
\begin{equation} \label{eq:IVThm1}
\mathbb{E}_{\beta, \Lambda, v}^{\IV}\left[ e^{\mOf{f}} \right]
\geq \exp\Big[\frac{1}{2(1+\varepsilon)\beta}\langle f, -\Delta_\Lambda^{-1}f\rangle\Big].
\end{equation}
Consequently, for such f,
\begin{equation} \label{eq:IVThm2}
\mathbb{E}_{\beta, \Lambda, v}^{\IV}\left[ \mOf{f}^2 \right] \geq \frac{1}{(1+\varepsilon)\beta}\langle f, -\Delta_\Lambda^{-1}f\rangle
\end{equation}
and, in particular, for some absolute constant $c>0$,
$$
\mathbb{E}_{\beta, \Lambda, v}^{\IV}\left[ \big(m_j-m_k\big)^2 \right] \geq \frac{c}{\beta}\cdot \log(\dist(j,k) + 1), \quad j,k\in\Lambda.
$$
\end{thm}

The results of Fr\"ohlich and Park \cite[Section 3]{upperBound} imply the following proposition, which shows that the fluctuations of the integer-valued discrete Gaussian free field are dominated by those of the real-valued discrete Gaussian free field. The proof is based on the methods of Ginibre \cite{Ginibre}, and is given in Section~\ref{sec:upper_bound} for completeness. Thus the lower bounds of Theorem~\ref{thm:IV} are almost tight at high temperature.
\begin{prop}\label{prop:IVUpperBound}
For any finite, connected graph $\Lambda$, $\beta > 0$, $v\in\Lambda$ and $f:\Lambda \rightarrow \mathbb{R}$ with $\sum_{j \in \Lambda} f_j = 0$,
\begin{equation*}
\mathbb{E}_{\beta, \Lambda, v}^{\IV}\left[ e^{\mOf{f}} \right] \leq
\mathbb{E}_{\beta, \Lambda, v}^{\GFF}\left[ e^{\phiOf{f}} \right] =
\exp\left({\frac{1}{2\beta} \langle f, -\Delta_\Lambda^{-1}f\rangle}\right).
\end{equation*}
\end{prop}

\subsection{Plane Rotator Model with Villain Interaction}
The \emph{plane rotator model with Villain interaction}, or simply the \emph{Villain model}, on the graph $\Lambda_L^{\zero}$ at inverse temperature $\beta>0$, normalized to equal $0$ at $\zv\in\Lambda_L^{\zero}$ (recalling that $\zv$ stands for the `wired boundary' of $\Lambda_L^{\zero}$), is the measure on $\theta:\Lambda_L^{\zero}\rightarrow[-\pi, \pi)$ given by
$$
d\mu_{\beta, \Lambda_L^{\zero}}^{\Vil}(\theta) :=
\frac{1}{Z_{\beta, \Lambda_L^{\zero}}^{\Vil}}
\prod_{j\thicksim \ell}
\sum_{m \in \mathbb{Z}} e^{-\frac{\beta}{2}(\theta_j - \theta_\ell + 2\pi m)^2}
\delta_{0}(\theta_{\zv})
\prod_{j \in \Lambda_L^{\zero}\setminus\{\zv\}} \mathds{1}_{[-\pi, \pi)}\big(\theta_j\big) d\theta_j
$$
where $d\theta_j$ is the Lebesgue measure on $\mathbb{R}$, and the normalization constant $Z_{\beta, \Lambda_L^{\zero}}^{\Vil}$ normalizes the measure $\mu_{\beta, \Lambda_L^{\zero}}^{\Vil}$ to be a probability measure. We write $\mathbb{E}_{\beta, \Lambda_L^{\zero}}^{\Vil}$ for the corresponding expectation.

McBryan and Spencer \cite{McSp} establish a quantitative version of the Mermin-Wagner theorem, showing that at all inverse temperatures $\beta>0$, for some absolute constants $C,c>0$,
$$
\mathbb{E}_{\beta, \Lambda_L^{\zero}}^{\Vil}\big[\cos \theta_j \big] \leq C(\dist(j, \zv)+1)^{-c/\beta}, \quad j\in\Lambda_L^{\zero}.
$$
High-temperature expansion shows that the above expectation, in fact, decays exponentially in $\dist(j, \zv)$ for small $\beta$ (see, e.g., \cite[Section 2.4]{ri}).
Thus, the following theorem of Fr\"ohlich and Spencer \cite{fs} proves the existence of a Berezinskii-Kosterlitz-Thouless transition in the Villain model.
\begin{thm} \label{thm:Vil}
There exist $\beta_0, C>0$ such that for every $L>1$ and every $\beta > \beta_0$,
$$
\mathbb{E}_{\beta, \Lambda_L^{\zero}}^{\Vil}\big[\cos \theta_j \big] \geq (\dist(j, \zv)+1)^{-C/\beta}, \quad j\in\Lambda_L^{\zero}.
$$
\end{thm}

\subsection{Discrete Gaussian Free Field with Periodic Single-Site Weights}\label{sec:DGFF_periodic_weights}
To prove Theorem \ref{thm:IV} and Theorem \ref{thm:Vil} we consider the \emph{discrete Gaussian free field with periodic single-site weights model}, defined in \eqref{label:eq116} below. This model is the dual of the lattice Coulomb gas model via the Sine-Gordon or Siegert representation, as discussed in \cite{fs}.

We say that $\lambda:\mathbb{R}\rightarrow\mathbb{R}$ is a \textit{real, even, normalized trigonometric polynomial} if
$$
\lambda(\phi) = 1 + 2\sum_{q=1}^N \hat{\lambda}(q)\cos(q\phi),
$$
for some integer $N > 0$ and real $(\hat{\lambda}(q))$, $1\leq q \leq N$. For notational convenience, we set $\hat{\lambda}(q)=0$ for $q > N$.
For a finite, connected graph $\Lambda$, $\beta>0$, $v\in\Lambda$ and
$$
\lambda_\Lambda:=(\lambda_j)_{j\in\Lambda}
$$
real, even, normalized trigonometric polynomials, we define a, not necessarily positive, measure on $\phi:\Lambda\to\mathbb{R}$ by
\begin{equation} \label{label:eq116}
d\mu_{\beta, \Lambda, \lambda_\Lambda, v}(\phi) := \frac{1}{Z_{\beta, \Lambda, \lambda_\Lambda, v}} \prod_{j \in \Lambda} \lambda_j(\phi_j)\midaLambdaV,
\end{equation}
where
$$
Z_{\beta, \Lambda, \lambda_\Lambda, v} := \int \prod_{j \in \Lambda} \lambda_j(\phi_j)\midaLambdaV.
$$
It will be shown in Theorem~\ref{thm234} below that $Z_{\beta, \Lambda, \lambda_\Lambda, v}>0$ under the conditions of our main theorem, Theorem~\ref{thm11}, so that $\mu_{\beta, \Lambda, \lambda_\Lambda, v}$ is well defined. Corresponding to our previous notation, we denote by $\mathbb{E}_{\beta, \Lambda, \lambda_\Lambda, v}$ the integration against $\mu_{\beta, \Lambda, \lambda_\Lambda, v}$ operation, bearing in mind that $\mu_{\beta, \Lambda, \lambda_\Lambda, v}$ may not be a positive measure. We note that $\mu_{\beta, \Lambda, \underbar 1, v} = \mu_{\text{\tiny$\beta, \Lambda, v$}}^{\GFF}$, where $\underbar 1$ denotes the vector of polynomials which are identically equal to $1$.

Our main interest is in computing $\mathbb{E}_{\beta, \Lambda, \lambda_\Lambda, v}\left[ F(\phi) \right]$, for functionals $F:\mathbb{R}^{\Lambda}\rightarrow \mathbb{R}$, which are invariant to translation by a constant function (i.e., $F(m) = F(m + c)$, for a constant function $c:\Lambda\rightarrow\mathbb{R}$). Our choice of normalization, the requirement that $\phi_v \in [-\pi, \pi)$, implies that for such functionals, integrals are independent of the choice of $v$,
$$
\mathbb{E}_{\beta, \Lambda, \lambda_\Lambda, v_1}\left[ F(\phi) \right] =
\mathbb{E}_{\beta, \Lambda, \lambda_\Lambda, v_2}\left[ F(\phi) \right], \quad v_1,v_2\in\Lambda.
$$

\subsection{Main Theorem} \label{sec:thm11}

\begin{definition}[Sub-Gaussian Condition]
We say that $\lambda: \mathbb{R} \rightarrow \mathbb{R}$ is $\SG$, if $\lambda$ is a real, even, normalized trigonometric polynomial, and
\begin{equation} \label{eqCondLambda}
\abs{\hat\lambda(q)} \leq \Gamma \cdot \exp\left[\left(\eta + \frac{\theta}{\beta}\right)q^2\right], \quad q\ge 1.
\end{equation}
\end{definition}

\begin{thm} \label{thm11}
For any $\Gamma > 0$, $\eta\in\mathbb{R}$, $0 \leq \theta < 1/16$, and $\varepsilon > 0$, there exists $\beta_0 > 0$ such that the following holds.
Let $\Lambda$ be a square domain, with free or periodic b.c., $v\in\Lambda$, $\lambda_\Lambda=(\lambda_j)_{j\in\Lambda}$ be a collection of $\SG$ polynomials, $\beta < \beta_0$, and $f:\Lambda \rightarrow \mathbb{R}$ satisfy $\sum_{j \in \Lambda} f_j = 0$. Then

\begin{equation} \label{label:eq54}
\mathbb{E}_{\beta, \Lambda, \lambda_\Lambda,v}\big[ e^{\phiOf{f}} \big] \geq
\exp\left[ \frac{1}{2(1+\varepsilon)\beta}\langle f, -\Delta_\Lambda^{-1}f\rangle \right].
\end{equation}
\end{thm}

We note that in the Gaussian free field case, when $\lambda_\Lambda = \underbar 1$, the lower bound in the right-hand side of \eqref{label:eq54} is close to the exact expression \eqref{eq:DGFFRes} for small $\beta$.

As Theorem \ref{thm11} holds uniformly for sub-Gaussian trigonometric polynomials $\lambda_j$, $j\in\Lambda$, with given parameters $\Gamma, \eta, \theta$, it follows that its statement may be extended to $2\pi$-periodic $\lambda_j:\mathbb{R}\rightarrow\mathbb{R}$, $j\in\Lambda$ (or $2\pi$-periodic measures) obtained as suitable limits of such polynomials.
 We demonstrate this technique by deducing Theorem~\ref{thm:IV} from Theorem~\ref{thm11}, in Section~\ref{sec:proofIV}.
One may also use this technique to change the normalization $\phi_v\in[-\pi, \pi)$ in Theorem~\ref{thm11} to other normalization choices including the standard normalization $\phi_v = 0$.

\subsection{Renormalization Step}

We state here the main step in the proof of Theorem \ref{thm11}, starting with a few definitions.

\subsubsection{Densities, Ensembles and Charges}

The \emph{support} of a function $f:\Lambda \rightarrow \mathbb{R}$ is
$$
\supp f := \{j\in\Lambda: f_j \neq 0 \}.
$$
A \textit{(charge) density} is a function $\varrho:\Lambda \rightarrow \mathbb{Z}$ which is not identically zero. Its diameter is
$$
d(\varrho) := \max_{i,j\in\supp\varrho} \dist(i,j).
$$
An \emph{ensemble} is a finite (possibly empty) collection of charge densities whose supports are mutually disjoint.
For each density $\varrho$, let $j\in\supp\varrho$ be such that there exists $k\in\supp\varrho$ with $d(\varrho)=\dist(j,k)$, with $j$ chosen in some fixed arbitrary way when more than one option is available. Define
\begin{equation} \label{eq:DRhoDef}
D(\varrho)=D_j(\varrho):=\{\ell\in\Lambda: \dist(\ell,j) < 2d(\varrho)\},
\end{equation}
and say that $j$ is the \emph{center} of $D(\varrho)$.
Note that $\supp\varrho \subset D(\varrho)$ for densities with $d(\varrho)\geq 1$.

The \textit{charge} $Q(\varrho)$ of a density $\varrho$ is defined by
$$
Q(\varrho) := \sum_{j\in\Lambda}\varrho(j).
$$
A density $\varrho$ is called \textit{neutral} if $Q(\varrho)=0$; otherwise it is said to be \textit{charged} or \textit{non-neutral}.
Observe that $d(\varrho)\geq1$ for neutral $\varrho$. We also denote
$$
\norm{\varrho}_2 := \sqrt{\sum_{j\in\Lambda} \varrho(j)^2}.
$$

We remark that the above terminology originates from the Coulomb gas representation of the model.

\subsubsection{Renormalization Step Theorem}


Let $f:\Lambda \rightarrow \mathbb{R}$ satisfy $\sum_{j \in \Lambda} f_j = 0$. The change of variables
$
\phi_j \rightarrow \phi_j + \sigma_j,
$
where $\sigma$ is defined in \eqref{eq:sigmaDef}, yields
\begin{equation} \label{eq:varChangeInE}
\begin{split}
\mathbb{E}_{\beta, \Lambda, \lambda_\Lambda, v}\big[ e^{\phiOf{f}} \big]
&=
\frac{1}{Z_{\beta, \Lambda, \lambda_\Lambda, v}} \int e^{\phiOf{f}}\prod_{j \in \Lambda} \lambda_j\big(\phi(j)\big)\midaLambdaV \\
&=
\frac{1}{Z_{\beta, \Lambda, \lambda_\Lambda, v}} \exp\left({\frac{1}{2\beta} \langle f, -\Delta_\Lambda^{-1}f\rangle}\right) \int \prod_{j \in \Lambda} \lambda_j\big(\phi(j) + \sigma(j)\big)\midaLambdaV.
\end{split}
\end{equation}
Our main challenge, therefore, is to analyse the integral on the right-hand side of \eqref{eq:varChangeInE}.
The next theorem expresses this integral as a convex combination of integrals arising from gradient fields with \emph{positive} interaction weights. This constitutes the main step in the proof of Theorem~\ref{thm11}.
\begin{thm} \label{thm234}
Let $\Gamma > 0$, $\eta\in\mathbb{R}$ and $0 \leq \theta < 1/16$.
There exist $\beta_0, c_1 > 0$ such that the following holds. Let $\Lambda$ be a square domain, with free or periodic b.c., let $v\in\Lambda$, let $\beta < \beta_0$, and let $\lambda_j: \mathbb{R} \rightarrow \mathbb{R}$, $j\in\Lambda$, be a collection of $\SG$ polynomials. Then there exist:
\begin{itemize}
  \item a finite collection of ensembles $\mathscr{F}$,
  \item positive $(c_\ens)$, $\ens\in\mathscr{F}$, summing to $1$,
  \item real $(\zed)$, $\varrho\in\ens$, $\ens\in\mathscr{F}$,
  \item functions $\bar{\varrho}:\Lambda\to\mathbb{R}$ for each $\varrho\in\ens$, $\ens\in\mathscr{F}$
\end{itemize}
such that for every $\sigma: \Lambda \rightarrow \mathbb{R}$,
\begin{equation}\label{eqObserve}
\int\prod_{j \in \Lambda} \lambda_j(\phi_j + \sigma_j)\midaLambdaV =
\sum_{\ens \in \mathscr{F}} c_\ens \int\prod_{\varrho \in \ens}[1 + \zed\cos(\phiOf{\bar{\varrho}} + \langle\sigma,\varrho\rangle)]\midaLambdaV,
\end{equation}
and the following properties are satisfied for every $\ens\in\mathscr{F}$:
\begin{enumerate}
\item All charge densities $\varrho\in\ens$ are neutral.
\item For every $\varrho\in\ens$,
\begin{equation} \label{eq:ZBoundThm234}
\abs{\zed} \leq \exp\left[-\frac{c_1}{\beta}\Big(\norm{\varrho}_2^2+\log_2(d(\varrho) + 1)\Big)\right].
\end{equation}
\item \label{enum:thm234PropD}
For distinct $\varrho_1, \varrho_2 \in \ens$, if $d(\varrho_1), d(\varrho_2) \in [2^k-1,2^{k+1}-2]$, $k\geq 1$, then ${D(\varrho_1) \cap D(\varrho_2) = \emptyset}$.
\end{enumerate}
\end{thm}
We remark that the proof also shows that for each $\varrho\in\ens$, $\ens\in\mathscr{F}$, the function $\bar{\varrho}$ is supported within $D(\varrho)$ and is neutral in the sense that $\sum_j \bar{\varrho}_j=0$, but this will not be used in the sequel.

The advantage of the representation given by \eqref{eqObserve} is that each term in the product on the right-hand side is explicitly positive, by \eqref{eq:ZBoundThm234}, and involves only the gradients of $\sigma$, since each $\varrho$ is neutral. Moreover, terms in which the density $\varrho$ has large diameter or norm are damped by \eqref{eq:ZBoundThm234} and thus, in a sense, the main contribution to the right-hand side comes from densities $\varrho$ with small diameter and norm.

In Section \ref{sec:proofOverview} we present an overview of the proof of Theorem \ref{thm234}, and in Sections \ref{sec:chap3}, \ref{sec:chap2}, \ref{sec:chap4} we elaborate on the more technical parts of the proof. Theorem~\ref{thm11} is deduced from Theorem~\ref{thm234} in Section~\ref{sec:proof11}.

\subsection{Main Differences from the Paper of Fr\"ohlich-Spencer} \label{sec:diffs_from_fs}
The proof presented in this paper is essentially the proof given by Fr\"ohlich and Spencer \cite{fs} for showing the delocalization of the two-dimensional integer-valued Gaussian free field at high temperature and the existence of the Berezinskii-Kosterlitz-Thouless (BKT) transition for the two-dimensional plane rotator model with Villain interaction.
The paper \cite{fs} includes additional results, omitted here, for other two-dimensional statistical physics models: Delocalization results are established for integer-valued random surfaces with certain non-Gaussian interaction functions, such as the Solid-On-Solid model, which allows to conclude the existence of the BKT transition for the standard plane rotator model (XY model) and for the $\mathbb{Z}_n$ clock models at intermediate temperature ranges.
The paper \cite{fs} also includes a discussion of the duality relation of the lattice Coulomb gas with measures of the type $\mu_{\beta, \Lambda, \lambda_\Lambda, v}$, defined in \eqref{label:eq116}, by the so-called Sine-Gordon or Siegert representation \cite{fr76, fs81, sig60}, and proves the absence of Debye screening in the Coulomb gas at \emph{fixed} activity $z$ and low temperature.
In this regard we mention that later works by Dimock-Hurd \cite{dh00}, Marchetti-Klein \cite{mk91} and Falco \cite{f12,f13} studied the behaviour of the lattice Coulomb gas at \emph{low} activity $z$ and extended the results of \cite{fs} up to the conjectured threshold $\beta = 8\pi$.

The proofs in \cite{fs} are conducted for the infinite-volume measures. For instance, when discussing random surfaces, one starts with the infinite-volume gradient Gibbs state of the discrete Gaussian free field and requires it to be integer-valued, or adds a periodic weight to it, on a finite subset of the lattice. In our treatment we chose to define all measures in finite volume (with the obtained estimates uniform in the volume) as we hope that this reduces the technical prerequisites required on the side of the reader. This choice leads to the less standard normalization condition that $\phi_v\in [-\pi, \pi)$ in our definition of the discrete Gaussian free field (Section~\ref{sec:DGFF}), but then has the benefit of allowing to eliminate the non-neutral densities (see Section~\ref{sec:removing_non_neutral}) produced by the renormalization construction in Theorem~\ref{thm:21} at an earlier stage than in \cite{fs}. In addition, we provide Theorem~\ref{thm11} in which different periodic weights are allowed at different lattice sites. This allows to pass to the more standard $\phi_v=0$ normalization and provides flexibility in taking infinite-volume limits (e.g., putting an integer-valued constraint on a finite subset of the lattice and removing all other weights in order to recover the setting of \cite{fs}).

We have included the somewhat technical statement of Theorem~\ref{thm234} already in the introduction, as it forms the central ingredient in the proof of the main theorem, with its proof occupying the bulk of this paper. Theorem~\ref{thm234} provides an \emph{equality} which, as such, may be of use in deducing further properties of the integer-valued, or otherwise periodically weighted, discrete Gaussian free field. We have also included a detailed proof overview for Theorem~\ref{thm234} in Section~\ref{sec:proofOverview}, which breaks the proof into several intermediate steps, in an attempt to make the argument more accessible.

Our treatment differs somewhat from \cite{fs} in the presentation and use of the spin wave functions. We have chosen to construct and highlight the properties of a single spin wave for each charge density (see Proposition~\ref{prop:aExistance}). The use of complex translation with the constructed spin waves in Section~\ref{sec:complex_translations_by_spin_waves} is then correspondingly modified from the treatment in \cite{fs}.

We have included the statement and proof of Proposition~\ref{prop:IVUpperBound}, which bounds the fluctuations of the integer-valued discrete Gaussian free field from \emph{above}, as it highlights the near-tightness of Theorem~\ref{thm:IV} at high temperature.

On the more technical side, a significant change to the notation is that our $\beta$ corresponds to $\frac{1}{\beta}$ in \cite{fs} in the discussion of random surface models. The reason for this is that the $\beta$ in \cite{fs} corresponds to the inverse temperature parameter in the Coulomb gas model which is in duality (which maps $\beta$ to $\frac{1}{\beta}$) with the discrete Gaussian free field with periodic weights. As our focus is on the latter model we preferred to use its natural inverse temperature notation. Our treatment of the random surface models in Theorem~\ref{thm:IV} and Theorem~\ref{thm11} allows for free or periodic boundary conditions. These are also the allowed boundary conditions in the main part of \cite{fs} but an explanation of the required modifications to handle zero boundary conditions appears there in Appendix D.

\section{Theorem \ref{thm234} - Proof Overview} \label{sec:proofOverview}

In this section we explain the derivation of Theorem~\ref{thm234}, proving some of the required statements, and leaving some for later sections.

Fix $\Lambda$ a square domain, with free or periodic b.c., $v\in\Lambda$, $\beta>0$, and $\lambda_j:\mathbb{R}\rightarrow\mathbb{R}$, $j\in\Lambda$ real, even, normalized trigonometric polynomials, for the rest of the section.

Fix also
\begin{equation} \label{eq:alphaMBounds}
3/2 < \alpha < 2, \quad M = 2^{16}.
\end{equation}
We remark that the value of $\alpha$ is important in works studying the behaviour of the Coulomb gas at low activity $z$, see Section~\ref{sec:diffs_from_fs}.

For two charge densities $\varrho_1, \varrho_2$, denote
$$
\dist(\varrho_1, \varrho_2) := \min_{\substack{i\in\supp\varrho_1 \\ j \in\supp\varrho_2}} \dist(i,j).
$$

\subsection{Square-Covering of Densities}

We start by defining a $2^k\times 2^k$ square, $k\geq 0$.
If $\Lambda = \Lambda_L^{\free}$, we say that $s\subset\Lambda_L^{\free}$ is a $2^k\times 2^k$ square if
$$
s = \{(c,d)\in \Lambda_L^{\free}: c-a, d-b \in \{0,1,\ldots,2^k-1 \} \} ,
$$
for some $(a,b)\in \Lambda_L^{\free}$.
If $\Lambda = \Lambda_L^{\per}$, we say that $s\subset\Lambda_L^{\per}$ is a $2^k\times 2^k$ square, if
$$
s = \{(c,d)\in \Lambda_L^{\per}: (c-a)\mod L, (d-b)\mod L \in \{0,1,\ldots,2^k-1 \} \},
$$
for some $(a,b)\in \Lambda_L^{\per}$. We also assume that a $2^k\times 2^k$ square $s$ contains $2^{k+1}$ points, unless $2^k > L$, in which case $s = \Lambda$.

We define a notion of a cover of a density $\varrho$, at different length scales. For any integer $k \geq 0$, let $\SqrsSet$ be a minimal collection of $2^k\times2^k$ squares covering the support of $\varrho$.
By minimal we mean that $\SqrsSet$ is chosen such that its cardinality, $\abs{\SqrsSet}$, is minimal.
The precise choice of $\SqrsSet$, when more than one minimal collection exists, is described later, in Section \ref{sec:proofOfSW}, and is only of relevance there.
Note that $\abs{\SqrsSetOf{0}{\varrho}} = \abs{\supp\varrho}$, and that $\abs{\SqrsSet} = 1$ for all $k \geq \log_2 (d(\varrho)+1)$.

Denote also
\begin{equation} \label{label:eq27}
A(\varrho) := \sum_{k=0}^{n(\varrho)} \abs{\SqrsSet}\text{ when $d(\varrho)\ge 1$ and $A(\varrho) := 0$ when $d(\varrho)=0$},
\end{equation}
where
\begin{equation} \label{label:eqNvarrho}
n(\varrho) := \lceil \log_2(M\cdot d(\varrho)^\alpha) \rceil.
\end{equation}

Define $\SqrsSetTag$, $k\geq 1$, as follows. If $\abs{\SqrsSet}=1$ then $\SqrsSetTag:=\emptyset$.
Otherwise define $\SqrsSetTag$ to be the sub-collection of those squares $s^\prime$ in $\SqrsSet$ which are separated from all other squares in $\SqrsSet$, as made precise by
\begin{equation} \label{label:eq33}
\SqrsSetTag := \big\{ s\in\SqrsSet: \dist(s^{\prime}, s)\geq 2M2^{\alpha (k + 1)} \ \ \ \forall s^{\prime}\in\SqrsSet \setminus \{s\} \big\}.
\end{equation}
Note that the cardinality $\abs{\SqrsSetTag}$ may depend on the exact choice of $\SqrsSet$, when more than one such minimal cover is possible.

The next proposition shows that the size of the cover $(\SqrsSet)$, as measured by $A(\varrho)$, is controlled by the sizes of $(\SqrsSetTag)$ and the diameter $d(\varrho)$.
\begin{prop} \label{prop:ABound}
There exists a positive absolute constant $\DOne$, such that for any density $\varrho$,
\begin{equation} \label{eq:ABounds}
\log_2(d(\varrho) + 1) \leq A(\varrho) \leq \DOne \cdot \left(\abs{\SqrsSetOf{0}{\varrho}} + \sum_{k=1}^{\infty} \abs{\SqrsSetTag}\right).
\end{equation}
\end{prop}

Note that the sum in the right-hand side of \eqref{eq:ABounds} is finite, since $\abs{\SqrsSet}=1$ for $k > \log_2 d(\varrho)$, and therefore $\abs{\SqrsSetTag}=0$ for such $k$.

\subsection{Expanding the Weights $\prod_{j \in \Lambda} \lambda_j(\phi(j) +\sigma(j))$ as a Convex Combination of Weights of Neutral Densities}

In the first step in the proof of Theorem \ref{thm234}, we use the following theorem, with $\psi$ replaced by $\phi+\sigma$.
\begin{thm} \label{thm:21}
There exists a positive absolute constant $\DTwo$, a finite collection of ensembles $\mathscr{F}$, positive $(c_\ens)$, $\ens\in\mathscr{F}$, summing to $1$ and real $(K(\varrho))$, $\varrho\in\ens\in\mathscr{F}$, such that for every $\psi:\Lambda\to\mathbb{R}$,
\begin{equation} \label{eq:thm21}
\prod_{j \in \Lambda} \lambda_j(\psi(j)) = \sum_{\ens \in \mathscr{F}} c_\ens \prod_{\varrho \in \ens} \big[ 1 + K(\varrho)\cos(\psiOf{\varrho}) \big]
\end{equation}
with the following properties satisfied for each $\ens \in \mathscr{F}$:
\begin{enumerate}[label={(\alph*)}]
\item \label{enum21:a}
There is at most one non-neutral density in $\ens$.
\item \label{enum21:b}
Every distinct densities $\varrho_1, \varrho_2 \in \ens$ satisfy
\begin{equation*}
\dist(\varrho_1, \varrho_2) \geq M[\min(d(\varrho_1), d(\varrho_2))]^\alpha.
\end{equation*}
Moreover, if $\ens$ contains a non-neutral density $\varrho_c$, then
$$
\dist(\varrho, \varrho_c) \geq M[d(\varrho)]^\alpha, \quad \varrho \in \ens, \varrho \neq \varrho_c.
$$
\item \label{enum21:c}
For each decomposition of a neutral density $\varrho\in\ens$ into $\varrho = \varrho_1 + \varrho_2$, where $\varrho_1, \varrho_2$ are densities with disjoint supports,
\begin{equation*}
\text{if }\dist(\varrho_1, \varrho_2) \geq 2M[\min(d(\varrho_1), d(\varrho_2))]^\alpha \text{ then $\varrho_1, \varrho_2$ are non-neutral}.
\end{equation*}

\item \label{enum21:d}
The coefficients $K(\varrho)$, appearing on the right-hand side of \eqref{eq:thm21}, satisfy for all neutral $\varrho\in\ens$
\begin{equation} \label{label:eq26}
0 \leq \abs{K(\varrho)} \leq e^{\DTwo \cdot A(\varrho)} \prod_{j \in \supp\varrho} e^{\varrho(j)^2} \big|\hat{\lambda_j}(\abs{\varrho(j)})\big|.
\end{equation}
\end{enumerate}
\end{thm}

We present the proof of Theorem \ref{thm:21} in Section \ref{sec:chap2}. We remark that the theorem remains valid (with the same proof) for any $\alpha > 1$, with the choice of $\mathscr{F}$ depending on $\alpha$.

Note that for each $\ens\in\mathscr{F}$, neutral $\varrho\in\ens$, and each $s\in\SqrsSetTag$, necessarily the density
\begin{equation} \label{eq:sCapVarrho}
(s\cap\varrho)(j):= \twopartdef{\varrho(j)}{j\in s}{0}{j\notin s}
\end{equation}
is non-neutral by \eqref{label:eq33} and property \ref{enum21:c} of Theorem~\ref{thm:21}. This property is part of the motivation for the definition of $\SqrsSetTag$.

\subsection{Measure Modification}

In this section we discuss the integral of terms of the form $\prod_{\varrho \in \ens}[1 + K(\varrho)\cos(\phiOf{\varrho} + \sigmaOf{\varrho})]$, appearing on the right-hand side of \eqref{eq:thm21}, when replacing $\psi$ by $\phi+\sigma$. Our goal is to modify measures of the form
$$
\prod_{\varrho \in \ens}[1 + K(\varrho)\cos(\phiOf{\varrho} + \sigmaOf{\varrho})]\midaLambdaV,
$$
into \emph{positive} measures with similar structure and same total mass, if $\beta$ is small.

\begin{thm} \label{thm:thm41}
There exists a positive absolute constant $\DThr$ such that the following holds. Let $\ens$ be an ensemble satisfying properties \ref{enum21:a}-\ref{enum21:c} of Theorem \ref{thm:21}, then for every $\sigma:\Lambda \rightarrow \mathbb{R}$,
\begin{equation} \label{eq:renomEq}
\begin{split}
&\int\prod_{\varrho \in \ens}[1 + K(\varrho)\cos(\phiOf{\varrho} + \langle\sigma, \varrho\rangle)]\midaLambdaV \\
&=
\int\prod_{\substack{\varrho \in \ens \\ Q(\varrho)=0}}[1 + \zed\cos(\phiOf{\varrho + \beta\Delta_\Lambda a_{\varrho,\ens}} + \langle\sigma,\varrho\rangle)]\midaLambdaV,
\end{split}
\end{equation}
with $\zed\in\mathbb{R}$ and $a_{\varrho,\ens}:\Lambda\to\mathbb{R}$ independent of $\sigma$, and
\begin{equation} \label{eqEloc}
\abs{\zed} \leq \abs{K(\varrho)}\exp\Big[-\frac{1}{\beta} \Big(\frac{1}{16}\norm{\varrho}_2^2 + \DThr\sum_{k=1}^\infty \abs{\SqrsSetTagOf{k}{\varrho}}\Big)\Big].
\end{equation}

Moreover, if $\lambda_j:\mathbb{R}\rightarrow\mathbb{R}$, $j\in\Lambda$ are $\SGS$ for some $\Gamma > 0$, $ \eta\in\mathbb{R}$, $0 \leq \theta < 1/16$,
and if $\ens$ satisfies \ref{enum21:d} of Theorem \ref{thm:21},
there exist constants $\beta_0,c_1>0$, depending only on $\Gamma, \eta, \theta$,
such that
\begin{equation} \label{label:eq53}
\abs{\zed} \leq \exp\left[-\frac{c_1}{\beta}\Big(\norm{\varrho}_2^2+\log_2(d(\varrho) + 1)\Big)\right], \quad 0 < \beta < \beta_0.
\end{equation}
\end{thm}

Note that Theorem~\ref{thm234} is an immediate consequence of Theorem~\ref{thm:21} and Theorem~\ref{thm:thm41}. In the rest of the section we discuss the proof of Theorem \ref{thm:thm41}.

To this end, we fix an ensemble $\ens$ satisfying properties \ref{enum21:a}-\ref{enum21:c} of Theorem \ref{thm:21}, and $\sigma:\Lambda \rightarrow \mathbb{R}$. We also denote for a density $\varrho$ and a function $a:\Lambda\rightarrow\mathbb{R}$
\begin{equation} \label{eq:Eloc}
E_\beta(a, \varrho) := \langle a_, \varrho \rangle - \frac{\beta}{2} \sum_{j\thicksim \ell} (a_j-a_\ell)^2 =
\langle a_, \varrho \rangle - \frac{\beta}{2} \langle a, -\Delta_\Lambda a \rangle.
\end{equation}

\subsubsection{Step 1: Removing the Non-Neutral Density}\label{sec:removing_non_neutral}

We first deal with the non-neutral density in $\ens$, if it exists, by showing that the integral on the left-hand side of \eqref{eq:renomEq} is unchanged by its removal. We remark that this is the part of the proof in which we make use of the restriction $\phi_v\in[-\pi, \pi)$. Denote the non-neutral density by $\varrho_c$ and write
\begin{equation} \label{eq:ens0}
\ens_0 := \ens \setminus \{ \varrho_c\}.
\end{equation}
Then
\begin{equation} \label{eq:chargedRem}
\begin{split}
&\int\prod_{\varrho \in \ens}[1 + K(\varrho)\cos(\phiOf{\varrho} + \sigmaOf{\varrho})]\midaLambdaV \\
&\ =
\int\prod_{\varrho \in \ens_0}[1 + K(\varrho)\cos(\phiOf{\varrho} + \sigmaOf{\varrho})]\midaLambdaV \\
&\ +
\int {K(\varrho_c)\cos(\phiOf{\varrho_c} + \sigmaOf{\varrho_c})\prod_{\varrho \in \ens_0}[1 + K(\varrho)\cos(\phiOf{\varrho} + \sigmaOf{\varrho})]\midaLambdaV}.
\end{split}
\end{equation}
We prove that the second term, in the last sum vanishes. To see this, we change variables in the left-hand side of \eqref{eq:chargedRem}, to
$$
\psi_k := \twopartdef{\phi_v}{k=v}{\phi_k-\phi_v}{k\neq v}.
$$
Since for every $\varrho\in\ens$
\begin{align*}
\phiOf{\varrho} &= \phi_v Q(\varrho) + \sum_{j\in\Lambda}(\phi_j-\phi_v)\varrho_j = \psi_v Q(\varrho) + \sum_{j\neq v} \psi_j\varrho_j, \\
\gradNormOld{\phi} &= \sum_{j\thicksim \ell} (\phi_j - \phi_\ell)^2 = \sum_{\substack{j\thicksim \ell \\ j,\ell \neq v}} (\psi_j - \psi_\ell)^2 + \sum_{j\thicksim v} \psi_j^2,
\end{align*}
and since $Q(\varrho_c)$ is an integer, it holds that
\begin{equation*}
\begin{split}
&\int K(\varrho_c)\cos(\phiOf{\varrho_c} + \sigmaOf{\varrho_c})\prod_{\varrho \in \ens_0}[1 + K(\varrho)\cos(\phiOf{\varrho} + \sigmaOf{\varrho})]\midaLambdaV \\
&=\int h_1(\psi) \left(\int_{-\pi}^{\pi} \cos[\psi_vQ(\varrho_c) + h_2(\psi)] \cdot d\psi_v\right)
\prod_{v\neq k\in\Lambda}d\psi_k = 0,
\end{split}
\end{equation*}
where $h_1, h_2$ are independent of $\psi_v$.

\subsubsection{Step 2: Constructing a Spin Wave for Neutral Densities} \label{sec:step2SW}

Let $\varrho\in\ens_0$ be some density, and consider the set of densities in $\ens_0$ of smaller or comparable diameter, defined by
\begin{equation} \label{eq:NRhoDef}
\ens_0(\varrho) := \{\varrho^\prime \in \ens_0: d(\varrho^\prime)\leq 2d(\varrho),\ \varrho \neq \varrho^\prime\}.
\end{equation}

We also denote
\begin{equation} \label{eq:D_plus_def}
D^{+}(\varrho) := \{ j\in\Lambda: \dist(j,D(\varrho)) \leq 1 \}.
\end{equation}

\begin{prop}[Spin Wave Existence] \label{prop:aExistance}
There exists a positive absolute constant $\DThr$ (independent of $\ens$), such that the following holds. For each $\varrho\in\ens_0$ there exists a function $a_{\varrho, \ens_0}: \Lambda\rightarrow\mathbb{R}$, denoted by $a_{\varrho}=a_{\varrho, \ens_0}$ for clarity of notation, with the following properties:
\begin{enumerate}
\item \label{enum1:aExistance}
For every $\varrho^\prime\in\ens_0(\varrho)$, $a_{\varrho}$ is constant on $D^{+}(\varrho^\prime)$.
\item \label{enum2:aExistance}
$\supp(a_{\varrho}) \subseteq D(\varrho)$. In particular, by \eqref{eq:alphaMBounds} and property \ref{enum21:b} of Theorem~\ref{thm:21}, $\supp a_{\varrho}$ and $\supp\varrho^\prime$ are disjoint for $\varrho^\prime\in\ens_0\setminus\ens_0(\varrho)$.
\item \label{enum3:aExistance}
$\supp(\Delta_\Lambda a_{\varrho}) \subset D(\varrho)$.
\item \label{enum4:aExistance}
\begin{equation} \label{eq:SpinWaveBound}
E_\beta(a_\varrho, \varrho) \geq
\frac{1}{\beta}\left( \frac{1}{16} \norm{\varrho}_2^2 + \DThr \sum_{k=1}^\infty \abs{\SqrsSetTagOf{k}{\varrho}} \right).
\end{equation}
\end{enumerate}
\end{prop}

We present the proof of Proposition \ref{prop:aExistance} in Section \ref{sec:chap4}. The construction requires $\alpha$ to be greater than $3/2$, and this is the only place where this condition is used.

\subsubsection{Step 3: Complex Translations by the Spin Waves}\label{sec:complex_translations_by_spin_waves}


The goal of this step is to show \eqref{eq:renomEq} with the condition \eqref{eqEloc}, where the functions $a_\varrho$ are the spin waves constructed in the previous step.

For each $\varrho\in\ens_0$, let $a_\varrho:\Lambda\rightarrow\mathbb{R}$ be the function given by Proposition \ref{prop:aExistance}.
We established in step 1 that
$$
\int\prod_{\varrho \in \ens}[1 + K(\varrho)\cos(\phiOf{\varrho} + \sigmaOf{\varrho})]\midaLambdaV =
\int\prod_{\varrho \in \ens_0}[1 + K(\varrho)\cos(\phiOf{\varrho} + \sigmaOf{\varrho})]\midaLambdaV.
$$

Rewrite the right-hand side of the last equation as
\begin{multline} \label{eq:multAllRho}
\int\prod_{\varrho \in \ens_0}[1 + K(\varrho)\cos(\phiOf{\varrho} + \sigmaOf{\varrho})]\midaLambdaV \\
=
\sum_{n:\ens_0\to\{-1,0,1\}}
\int \prod_{\varrho\in\ens_0}
\left(\frac{K(\varrho)}{2}\right)^{\abs{n(\varrho)}}
e^{i\cdot n(\varrho)\cdot(\phiOf{\varrho} + \sigmaOf{\varrho})}\midaLambdaV \\
=
\sum_{n:\ens_0\to\{-1,0,1\}}
\left[\prod_{\varrho\in\ens_0}
\left(\frac{K(\varrho)}{2}\right)^{\abs{n(\varrho)}}\right]
\int
\exp\Big[ i\cdot \Big(\Big\langle\phi,\sum_{\varrho\in\ens_0}n(\varrho)\cdot\varrho\Big\rangle + \Big\langle\sigma,\sum_{\varrho\in\ens_0}n(\varrho)\cdot\varrho\Big\rangle\Big)\Big]\midaLambdaV.
\end{multline}
We apply the complex translation method of McBryan and Spencer \cite{McSp} on every term of the right-hand side of \eqref{eq:multAllRho}. The method in our case is based on the following equality, which is a consequence of a change of integration contour,
\begin{equation}\label{eq:ComplexTrans}
\int e^{i\phiOf{\uptau}}\midaLambdaV = e^{-E_\beta(a,\uptau)}  \int e^{i\phiOf{\uptau+\beta\Delta_\Lambda a}} \midaLambdaV,
\quad \uptau, a:\Lambda\to\mathbb{R},
\quad \sum_{j\in\Lambda} \uptau_j=0.
\end{equation}
For completeness, we detail the proof of \eqref{eq:ComplexTrans} in Section \ref{sec:complex}. We remark that Fr\"ohlich and Spencer \cite[Lemma 4.3, Section 6 and Section 7]{fs} apply complex translation also to certain non-Gaussian measures, which requires a more general version of \eqref{eq:ComplexTrans}. These non-Gaussian measures arise when analyzing the plane rotator model with XY interaction, $\mathbb{Z}_n$ clock models and the Solid-On-Solid models.

Let $n:\ens_0\to\{-1,0,1\}$ and
$$
a_n:=\sum_{\varrho\in\ens_0} n(\varrho)\cdot a_\varrho \quad
\uptau_n := \sum_{\varrho\in\ens_0}n(\varrho)\cdot\varrho
$$
Since, by the properties of Proposition~\ref{prop:aExistance}, it holds that
$$
\langle a_{\varrho}, \varrho'\rangle=0, \quad \langle a_{\varrho}, \Delta_\Lambda a_{\varrho'} \rangle=0, \quad \varrho, \varrho'\in\ens, \varrho \neq \varrho',
$$
we obtain,
$$
E_\beta(a_n, \uptau_n) = \sum_{\varrho\in\ens_0}\sum_{\varrho'\in\ens_0}n(\varrho)n(\varrho')\langle a_\varrho, \varrho' \rangle - \frac{\beta}{2}\sum_{\varrho\in\ens_0}\sum_{\varrho'\in\ens_0}n(\varrho)n(\varrho')\langle a_{\varrho}, -\Delta_\Lambda a_{\varrho'} \rangle = \sum_{\varrho\in\ens_0} \abs{n(\varrho)} E_\beta(a_\varrho, \varrho).
$$
Therefore, using \eqref{eq:ComplexTrans},
\begin{multline}
\left[\prod_{\varrho\in\ens_0}
\left(\frac{K(\varrho)}{2}\right)^{\abs{n(\varrho)}}\right]
\int
\exp\Big[ i\cdot \Big(\Big\langle\phi,\sum_{\varrho\in\ens_0}n(\varrho)\cdot\varrho\Big\rangle + \Big\langle\sigma,\sum_{\varrho\in\ens_0}n(\varrho)\cdot\varrho\Big\rangle\Big)\Big]\midaLambdaV \\
=
\int
\prod_{\varrho\in\ens_0}
\left(\frac{K(\varrho)}{2} e^{-E_\beta(a_\varrho, \varrho)} \right)^{\abs{n(\varrho)}}
\exp\Big[ i\cdot n(\varrho)\cdot \big(\langle\phi,\varrho + \beta\Delta_\Lambda a_\varrho\rangle + \langle\sigma,\varrho\rangle\big)\Big]\midaLambdaV,
\end{multline}
which finishes the proof of \eqref{eq:renomEq} with the condition \eqref{eqEloc}, for
\begin{equation*}
  \zed := K(\varrho)\cdot
\exp\Big(-E_\beta(a_{\varrho}, \varrho)\Big).
\end{equation*}

\subsubsection{Step 4: Bounding the Coefficients $\zed$}
In this step we finish the proof, by showing \eqref{label:eq53}. The proof uses \eqref{label:eq26} and \eqref{eqEloc}, together with the sub-Gaussian condition.

We start by using $\theta<\frac{1}{16}$, and rewriting \eqref{eq:ABounds} as
$$
A(\varrho) \leq \DOne \abs{\SqrsSetOf{0}{\varrho}} + \gamma \DOne \sum_{k=1}^{\infty} \abs{\SqrsSetTag},
$$
where $\gamma > 1$ is large enough so that
\begin{equation} \label{eq:D2Choice}
\frac{1}{16} - \frac{\DThr}{\gamma} - \theta > 0.
\end{equation}
Using \eqref{eqEloc}, \eqref{eqCondLambda} and \eqref{label:eq26} we obtain
$$
\abs{\zed} \leq
\exp\left[
\DTwo \cdot A(\varrho) + \left(\ln(\Gamma) + \eta - \left(\frac{1}{16}-\theta\right)\frac{1}{\beta}\right)\norm{\varrho}_2^2
 - \frac{\DThr}{\beta}\sum_{k=1}^\infty \abs{\SqrsSetTagOf{k}{\varrho}}\Big)
\right].
$$
Applying the bounds on $A(\varrho)$ (Proposition \ref{prop:ABound}), and assuming that $\beta$ small enough, yields
$$
\abs{\zed} \leq
\exp\left[
\left(\ln(\Gamma) + \eta - \left(\frac{1}{16}-\frac{\DThr}{\gamma}-\theta\right)\frac{1}{\beta}\right)\norm{\varrho}_2^2
- \left(\frac{\DThr}{\gamma \DOne \beta} - \DTwo\right)\log_2(d(\varrho)+1)
\right],
$$
which finishes the proof, if $\beta$ is small enough.

\section{Proof of Theorem~\ref{thm11}} \label{sec:proof11}
In this section we deduce Theorem~\ref{thm11} from Theorem~\ref{thm234}. To this end, let $\Lambda$ be a square domain with free or periodic b.c., let $v\in\Lambda$, let $\Gamma >0$, $\eta\in\mathbb{R}$, $0\leq\theta < 1/16$, let $\lambda_\Lambda=(\lambda_j)_{j\in\Lambda}$ be a collection of $\SG$ functions, let $\varepsilon > 0$ and let $f:\Lambda \rightarrow \mathbb{R}$ satisfy $\sum_{j \in \Lambda} f_j = 0$.


Using \eqref{eq:sigmaDef}, \eqref{eq:varChangeInE}, we obtain
\begin{equation} \label{label:eq513}
\mathbb{E}_{\beta, \Lambda, \lambda_\Lambda, v}\big[ e^{\phiOf{f}} \big] = Z_{\beta, \Lambda, \lambda_\Lambda, v}^{-1}
\exp\left(\frac{1}{2\beta} \langle f, -\Delta_\Lambda^{-1}f\rangle\right)
Z_{\beta, \Lambda, \lambda_\Lambda, v}(\sigma),
\end{equation}
where
\begin{equation}\label{eq:ZOfSigmaDef}
Z_{\beta, \Lambda, \lambda_\Lambda, v}(\sigma):= \int\prod_{j\in\Lambda}\lambda_j(\phi_j + \sigma_j)\midaLambdaV.
\end{equation}
Note that since $Z_{\beta, \Lambda, \lambda_\Lambda, v}(0) = Z_{\beta, \Lambda, \lambda_\Lambda, v}$, our goal is to estimate
\begin{equation} \label{eq:z_ratio}
Z_{\beta, \Lambda, \lambda_\Lambda, v}(\sigma) / Z_{\beta, \Lambda, \lambda_\Lambda, v}(0).
\end{equation}

Let $\beta_0$ be the positive number given by Theorem~\ref{thm234}. The equality in that theorem states that
$$
Z_{\beta, \Lambda, \lambda_\Lambda, v}(\sigma) = \sum_{\ens \in \mathscr{F}} c_\ens Z_\ens (\sigma), \quad \beta < \beta_0
$$
where
\begin{equation}\label{eq:z_ens_def}
Z_\ens(\sigma) := \int \prod_{\varrho \in \ens} \big[ 1 + \zed\cos(\phiOf{\bar{\varrho}} + \sigmaOf{\varrho}) \big] \midaLambdaV.
\end{equation}

To estimate the ratio \eqref{eq:z_ratio}, we show that there exists $\beta_1>0$ such that for every $\ens\in\mathscr{F}$,
\begin{equation} \label{label:eq524}
\frac{Z_\ens(\sigma)}{Z_\ens(0)} \geq \exp\Big(-\frac{\varepsilon\beta}{2(1+\varepsilon)} \gradNorm{\sigma}\Big), \quad \beta<\beta_1.
\end{equation}
This allows us to continue \eqref{label:eq513} using the fact that
\begin{equation*} \label{label:eq525}
\gradNorm{\sigma} = \gradNormOld{\sigma} = (1/\beta^2) \langle f, -\Delta_\Lambda^{-1}f\rangle,
\end{equation*}
and that the right-hand side of \eqref{label:eq524} is independent of $\ens$, and to obtain for $\beta < \beta_1$
\begin{equation*}
\mathbb{E}_{\beta, \Lambda, \lambda_\Lambda, v}\big[ e^{\phiOf{f}} \big] \geq
\exp \Big[ \frac{1}{2\beta} \langle f, -\Delta_\Lambda^{-1}f\rangle -\frac{\varepsilon}{2(1+\varepsilon)\beta} \langle f, -\Delta_\Lambda^{-1}f\rangle\Big] =
\exp \Big[ \frac{1}{2(1+\varepsilon)\beta} \langle f, -\Delta_\Lambda^{-1}f\rangle \Big].
\end{equation*}
This is the statement of Theorem~\ref{thm11}, with $\beta_1$ playing the role of $\beta_0$. It remains to prove \eqref{label:eq524}, which we deduce from the following claims.
\begin{claim}  \label{lemma:doubleCos}
Let $x,y\in\mathbb{R}$ and $0 < \abs{z} < 1/8$. There exists an absolute constant $\DFr>0$ such that
\begin{equation} \label{newlemma521eq2pi}
1+z\cos(x+y) \geq \exp\left(-\frac{z\sin x\sin y}{1+z\cos x}-\DFr \abs{z} y^2\right) (1+z\cos x).
\end{equation}
\end{claim}

\begin{claim} \label{newlemma55eq}
Let $D>0$. There exists $0 < \beta_2 \leq \beta_0$ such that
\begin{equation}
\sum_{\varrho \in \ens} \abs{\zed} \cdot \sigmaOf{\varrho}^2 \leq
\frac{\beta}{D} \gradNorm{\sigma},\quad\beta<\beta_2, \quad \ens\in\mathscr{F}.
\end{equation}
\end{claim}

Let us first deduce \eqref{label:eq524} and then prove the claims. Fix $\ens\in\mathscr{F}$. Choose $\beta_3$ using \eqref{eq:ZBoundThm234} so that $|\zed|<\frac{1}{8}$ for all $\varrho\in\ens$ and $\beta < \beta_3$.
We therefore obtain by Claim~\ref{lemma:doubleCos} that
\begin{equation*}
\frac{Z_\ens(\sigma)}{Z_\ens (0)}
\geq
\exp \big[-\DFr\sum_{\varrho \in \ens} \abs{\zed}\sigmaOf{\varrho}^2 \big] \int e^{S(\ens, \phi)}Z_\ens(0)^{-1}[1 + \zed\cos\phiOf{\bar{\varrho}}]\midaLambdaV
\end{equation*}
for $\beta<\beta_3$, where $S(\ens, \phi)$ is defined to be the odd function (in $\phi$)
$$
S(\ens, \phi) := -\sum_{\varrho\in\ens} \frac{\zed \sin\phiOf{\bar{\varrho}} \sin\sigmaOf{\varrho}}{1+\zed \cos\phiOf{\bar{\varrho}}},
$$
and where we note that $Z_\ens(0)^{-1}[1 + \zed\cos\phiOf{\bar{\varrho}}]\midaLambdaV$ defines a probability measure, invariant under the mapping $\phi\mapsto-\phi$. Jensen's inequality and Claim~\ref{newlemma55eq} applied with $D:=\frac{2(1+\varepsilon)}{\varepsilon} \DFr$ now imply that
\begin{equation*}
\frac{Z_\ens(\sigma)}{Z_\ens (0)}
\geq \exp \big(-\DFr\sum_{\varrho \in \ens} \abs{\zed}\sigmaOf{\varrho}^2 \big)
\geq
\exp \left(-\frac{\varepsilon\beta}{2(1+\varepsilon)}\gradNorm{\sigma}\right),\quad \beta<\min(\beta_2, \beta_3),
\end{equation*}
establishing \eqref{label:eq524} with $\beta_1:=\min(\beta_2, \beta_3)$.

\begin{proof}[Proof of Claim~\ref{lemma:doubleCos}]

\begin{multline*}
1 + z\cos(x+y)
= 1 + z(\cos x \cos y - z\sin x \sin y)\\
= (1+z\cos x)\left( 1+ \frac{z}{1+z\cos x} (\cos x (\cos y - 1)-\sin x \sin y)\right).
\end{multline*}
Now, using that
\begin{gather*}
\left| \frac{z}{1+z\cos x} (\cos x (\cos y - 1)-\sin x \sin y) \right| \leq \frac{3}{7}, \quad x,y,z\in\mathbb{R}, \abs{z}<\frac{1}{8},\\
\ln(1+t) \geq t-2t^2, \quad \abs{t}\leq \frac{1}{2},\\
\abs{\sin t} \leq \abs{t}, \quad \cos t - 1 \geq -\frac{1}{2}t^2, \quad t \in \mathbb{R},\\
(s-t)^2 \leq 2(s^2+t^2), \quad s,t \in \mathbb{R},
\end{gather*}
we obtain
\begin{multline*}
\ln\left( 1+ \frac{z}{1+z\cos x} (\cos x (\cos y - 1)-\sin x \sin y)\right) \\
\geq
\frac{z}{1+z\cos x} (\cos x (\cos y - 1)-\sin x \sin y) -
\frac{4z^2}{(1+z\cos x)^2} (\cos^2 x (\cos y - 1)^2+\sin^2 x \sin^2 y)\\
\geq
-\frac{z\sin x\sin y}{1+z\cos x}-\DFr \abs{z} y^2. \qedhere
\end{multline*}
\end{proof}

\begin{proof}[Proof of Claim~\ref{newlemma55eq}]
Fix $\ens\in\mathscr{F}$ and let $\varrho\in\ens$. The fact that $\varrho$ takes integer values and is neutral allows us to express $\sigmaOf{\varrho}$ as
\begin{equation*}
\sigmaOf{\varrho} = \sum_{\substack{j,\ell\in D(\varrho)\\j\thicksim\ell}} c_{\{j,\ell\}}(\sigma_j - \sigma_\ell)
\end{equation*}
for integer $c_{\{j,\ell\}}$ satisfying $|c_{\{j,\ell\}}|\le \frac{1}{2} \sum_j\abs{\varrho(j)}\leq\frac{1}{2}\norm{\varrho}_2^2$ (this may be seen by induction on $\|\varrho\|_1:=\sum_j\abs{\varrho(j)}$. The case that $\|\varrho\|_1=2$ is clear and otherwise we can decompose $\varrho=\varrho_1+\varrho_2$ for neutral densities $\varrho_1, \varrho_2$ supported in $\supp(\varrho)$ and having $\|\varrho_1\|_1,\|\varrho_2\|_1<\|\varrho\|_1$). It follows by Cauchy-Schwarz that
\begin{equation*}
|\sigmaOf{\varrho}|^2 \le |D(\varrho)|\cdot\norm{\varrho}_2^4\sum_{\substack{j,\ell\in D(\varrho)\\j\thicksim\ell}}(\sigma_j - \sigma_\ell)^2.
\end{equation*}
Now, using \eqref{eq:ZBoundThm234} and setting $0<\beta_2<\beta_0$ to be small enough, we obtain for each $\beta < \beta_2$ that
\begin{equation*}
\abs{\zed} \cdot \sigmaOf{\varrho}^2 \leq
\frac{\beta}{D\cdot (d(\varrho)+1)}\sum_{\substack{j,\ell\in D(\varrho)\\j\thicksim\ell}}(\sigma_j - \sigma_\ell)^2.
\end{equation*}
We use now property \ref{enum:thm234PropD} of Theorem~\ref{thm234}, which states that $D(\varrho_1)$, $D(\varrho_2)$ are disjoint whenever $\varrho_1, \varrho_2 \in \ens$ are distinct densities satisfying that $d(\varrho_1)+1, d(\varrho_2)+1 \in [2^k,2^{k+1})$ for some $k\geq 1$, and conclude that
\begin{equation*}
\sum_{\varrho\in\ens}\abs{\zed} \cdot \sigmaOf{\varrho}^2 \leq \frac{\beta}{D}\sum_{k\ge 1} 2^{-k}\sum_{\substack{\varrho\in \ens\\2^k\le d(\varrho)+1<2^{k+1}}}\sum_{\substack{j,\ell\in D(\varrho)\\j\thicksim\ell}}(\sigma_j - \sigma_\ell)^2\le \frac{\beta}{D}\gradNorm{\sigma}.\qedhere
\end{equation*}
\end{proof}

\section{Completing the Proof of Theorem~\ref{thm234}}
In this section we prove Proposition~\ref{prop:ABound}, Theorem~\ref{thm:21}, Proposition~\ref{prop:aExistance} and \eqref{eq:ComplexTrans}. We fix $\Lambda$ to be a square domain with free or periodic b.c. and $\beta>0$ throughout the section.


\subsection{Proof of Proposition~\ref{prop:ABound}}\label{sec:chap3}


Fix a density $\varrho$. The proposition is clear when $d(\varrho)=0$ so we assume that $d(\varrho)\ge 1$. The lower bound follows simply from the fact that
\begin{equation*}
  A(\varrho) = \sum_{k=0}^{n(\varrho)} \abs{\SqrsSet} \ge n(\varrho) + 1 = \lceil \log_2(M\cdot d(\varrho)^\alpha)\rceil + 1 \ge \log_2(d(\varrho)+1).
\end{equation*}
The rest of the section is devoted to the proof of the upper bound. Set $b:=\alpha+3+\log_2(M)$ and define
\begin{equation*}
  \gamma(k):=\lfloor \alpha^{-1}(k-b)\rfloor\quad \text{for }k\ge b\text{ integer}.
\end{equation*}

The next lemma allows us to recursively bound $\abs{\SqrsSet}$.

\begin{lemma} \label{lemma:GammaReason}
Let $k\ge b$ be an integer. If $\abs{\SqrsSetGamma} \geq 2$ then
\begin{equation} \label{label:GammaRec}
\abs{\SqrsSet} \leq \frac{1}{2} \abs{\SqrsSetGamma} + \abs{\SqrsSetTagGamma}.
\end{equation}
\end{lemma}

We need a simple preparatory claim.
\begin{claim}\label{cl:graph_cover}
  Let $G = (V(G), E(G))$ be a finite, non-empty simple graph with no isolated vertices. Then $V(G)$ is the union of $\lfloor|V(G)|/2\rfloor$ sets, each of which has size $2$ or $3$ and is connected in $G$.
\end{claim}
\begin{proof}
  It suffices to prove the claim when $G$ is a tree. The proof then proceeds by induction on $|V(G)|$. The claim is clear if $|V(G)|\le 3$. Otherwise, necessarily one has three distinct vertices $v_1,v_2,v_3\in V(G)$ with $v_1$ of degree $1$, $\{v_1, v_2\}, \{v_2, v_3\}\in E(G)$ and either $v_2$ has degree $2$ or $v_3$ has degree $1$. In both cases we take $\{v_1, v_2, v_3\}$ as one of the covering sets of $V(G)$. In the first case we then proceed, using the induction step, after erasing $v_1, v_2$ from $G$ and in the second case after erasing $v_1, v_3$.
\end{proof}

\begin{proof}[Proof of Lemma~\ref{lemma:GammaReason}]
We will show that
\begin{equation} \label{label:lab38}
\abs{\SqrsSet} \leq
\frac{1}{2}\abs{\SqrsSetDtagGamma} + \abs{\SqrsSetTagGamma}
\end{equation}
from which the lemma follows as $\abs{\SqrsSetDtagGamma}\le \abs{\SqrsSetGamma}$ by definition. We first note that $\varrho$ is covered by $\SqrsSetTagGamma\cup\SqrsSetDtagGamma$. The relation \eqref{label:lab38} follows if $\SqrsSetDtagGamma=\emptyset$ and otherwise will follow by showing that $\SqrsSetDtagGamma$ can be covered by at most $\frac{1}{2}\abs{\SqrsSetDtagGamma}$ squares of side length $2^k$. To see the last assertion, consider a graph $G$ in which the vertex set is $\SqrsSetDtagGamma$ and two squares $s_1, s_2$ are adjacent if $\dist(s_1, s_2) < 2^{\alpha \gamma(k) + b} = 2M2^{\alpha(\gamma(k)+1)}$. The definitions of $\SqrsSetDtagGamma$, $\SqrsSetTagGamma$ and the assumption $\abs{\SqrsSetGamma} \geq 2$ imply that there are no isolated vertices in $G$. Applying Claim~\ref{cl:graph_cover} to $G$ reduces our task to showing that if $\{s_1, s_2, s_3\}\subseteq V(G)$ is connected then $s_1, s_2, s_3$ may be covered by a single square of side length $2^k$. Indeed, suppose $\dist(s_1, s_2), \dist(s_2, s_3)< 2^{\alpha \gamma(k) + b}$. Then
\begin{equation*}
  d(s_1\cup s_2\cup s_3)\le d(s_1) + d(s_2) + d(s_3) + \dist(s_1, s_2) + \dist(s_2, s_3)  \leq 3\cdot 2^{\gamma(k)+1} + 2\cdot 2^{\alpha \gamma(k) + b-2} < 2^k
\end{equation*}
finishing the proof of the lemma.
\end{proof}

In order to apply \eqref{label:GammaRec} we first check that $\abs{\SqrsSetGamma} \geq 2$ for $b\leq k\leq n(\varrho)$. Indeed,
\begin{equation*}
\gamma(k) \leq \alpha^{-1}(k-b) \leq
\alpha^{-1}(\alpha \log_2d(\varrho) -\alpha -2)\leq
\log_2(d(\varrho)+1)-1.
\end{equation*}

For each $k\geq b$, we now iterate inequality \eqref{label:GammaRec} $\ell(k)$ times, where $\ell(k)$ is the maximal number for which $\gamma^{\ell(k)}(k) \geq 0$ ($\gamma^m$ denotes the m-fold composition of $\gamma$ with itself). Thus,
\begin{equation} \label{label:lab310}
\abs{\SqrsSet} \leq
2^{-\ell(k)}
\abs{\SqrsSetOf{\gamma^{\ell(k)}(k)}{\varrho}} + \sum_{m=0}^{\ell(k)-1} 2^{-m} \abs{\SqrsSetTagOf{\gamma^{m+1}(k)}{\varrho}} \leq
2^{-\ell(k)}\abs{\SqrsSetOf{0}{\varrho}} + \sum_{m=0}^{\ell(k)-1} 2^{-m} \abs{\SqrsSetTagOf{\gamma^{m+1}(k)}{\varrho}}.
\end{equation}

We now estimate $\ell(k)$. One checks by a simple induction that for $k\geq b$,
\begin{equation} \label{eq:gamma_m_bounds}
\alpha^{-m}k - \frac{\alpha+b}{\alpha-1}
\leq
\alpha^{-m}k - b\sum_{j=1}^{m}\alpha^{-j}-\sum_{j=0}^{m-1}\alpha^{-j}
\leq
\gamma^m(k)
\leq
\alpha^{-m}k - b\sum_{j=1}^m\alpha^{-j}, \quad 0\leq m\leq \ell(k).
\end{equation}

Set $k_0:=\frac{\alpha+b}{\alpha-1}$. The lower bound in \eqref{eq:gamma_m_bounds} immediately implies that
\newcommand*{\alphadelta}{(\alpha-1)^{-1}(\alpha+b)}
\begin{equation} \label{label:eq314}
\ell(k) \geq \twopartdef { 0 } {0 \leq k < k_0} {{\frac{\log_2(k/k_0)}{\log_2(\alpha)}}} {\text{otherwise}}.
\end{equation}

Let $m, j \geq 0$. We estimate the cardinality $\abs{N_{m,j}}$ of the sets
$$
N_{m,j} := \{k: \gamma^m(k) = j\}.
$$
Let $k_-$ be the minimal and $k_+$ the maximal integer in $N_{j,m}$. Then by \eqref{eq:gamma_m_bounds}
\begin{equation} \label{label:lab316}
\abs{N_{m,j}} = k_+ - k_- + 1 \leq \alpha^m \sum_{\ell=0}^{m-1}\alpha^{-\ell} + 1 \leq \alpha^m \sum_{\ell=0}^{\infty}\alpha^{-\ell} + 1 = \alpha^m \frac{\alpha}{\alpha-1} + 1 \leq \alpha^m \frac{2\alpha}{\alpha-1}.
\end{equation}

Now, using \eqref{label:lab310} we have
\begin{equation} \label{eq:ABound0}
A(\varrho) = \sum_{k=0}^{n(\varrho)}\abs{\SqrsSet} \leq
\abs{\SqrsSetOf{0}{\varrho}}\sum_{k=0}^{n(\varrho)} 2^{-\ell(k)} +
\sum_{k=0}^{n(\varrho)}\sum_{m=0}^{\ell(k)-1} 2^{-m} \abs{\SqrsSetTagOf{\gamma^{m+1}(k)}{\varrho}} .
\end{equation}

We continue by bounding both terms. First note that by \eqref{eq:alphaMBounds} and \eqref{label:eq314} 
\begin{equation} \label{eq:ABound1}
\sum_{k=0}^{n(\varrho)} 2^{-\ell(k)} \leq \sum_{k=0}^{\infty} 2^{-\ell(k)} \leq k_0 + \sum_{k=k_0+1}^{\infty} 2^{-\ell(k)} \leq k_0 + \sum_{k=k_0+1}^{\infty}(k_0/k)^{1/\log_2(\alpha)} < \infty.
\end{equation}

Let $k(\varrho)$ be the largest k for which $\SqrsSetTag \neq \emptyset$, so using \eqref{eq:alphaMBounds} and \eqref{label:lab316} 
\begin{multline} \label{eq:ABound2}
\sum_{k=0}^{n(\varrho)}\sum_{m=0}^{\ell(k)-1} 2^{-m} \abs{\SqrsSetTagOf{\gamma^{m+1}(k)}{\varrho}}
=
\sum_{j=0}^{k(\varrho)} \bigg( \sum_{k=0}^{n(\varrho)}\sum_{m=0}^{\ell(k)-1} 2^{-m} \delta_{\gamma^{m+1}(k), j} \bigg) \abs{\SqrsSetTagOf{j}{\varrho}} \\
\leq
\sum_{j=0}^{k(\varrho)} \bigg( \sum_{m,k=0}^{\infty} 2^{-m} \delta_{\gamma^{m+1}(k), j} \bigg) \abs{\SqrsSetTagOf{j}{\varrho}}
=
\sum_{j=0}^{k(\varrho)} \bigg( \sum_{m=0}^{\infty} 2^{-m} \abs{N_{m+1,j}} \bigg) \abs{\SqrsSetTagOf{j}{\varrho}} \\
\leq
\sum_{j=0}^{k(\varrho)} \bigg( \alpha \frac{2\alpha}{\alpha-1} \sum_{m=0}^{\infty} \big( \frac{\alpha}{2} \big) ^m \bigg) \abs{\SqrsSetTagOf{j}{\varrho}}
\leq
\frac{4\alpha^2}{(\alpha-1)(2-\alpha)}\sum_{j=0}^{\infty} \abs{\SqrsSetTagOf{j}{\varrho}}.
\end{multline}
The upper bound in \eqref{eq:ABounds} follows by substituting \eqref{eq:ABound1} and \eqref{eq:ABound2} in \eqref{eq:ABound0}.

\subsection{Decomposition of the Weights as a Mixture of Ensembles}\label{sec:chap2}

In this section we prove Theorem~\ref{thm:21}.

\subsubsection{Additional Notation}
A charge density $\varrho_1$ is said to be \textit{compatible} with an ensemble $\ensE$ if
$$
\varrho_1 = \sum_{\varrho \in \ensE} \varepsilon(\varrho_1, \varrho)\varrho, \ \ \ \text{with} \ \ \ \varepsilon(\varrho_1, \varrho) \in \{-1,0,1\}.
$$
Note that the $\varepsilon(\varrho_1, \varrho)$ are uniquely determined, as the densities in an ensemble have disjoint supports.

We say that an ensemble $\ensE_1$ is a \textit{parent} of an ensemble $\ensE_2$, and write $\ensE_1 \rightarrow \ensE_2$, when every charge density $\varrho \in \ensE_2$ is compatible with $\ensE_1$. We say that a density $\varrho$ is a \textit{constituent} of a density $\varrho_1$, and write $\varrho \subset \varrho_1$, when $\supp\varrho \subset \supp\varrho_1$ and $\varrho(j) = \varrho_1(j)$ for all $j \in \supp\varrho$.

For an integer $k\ge-1$, an ensemble $\ensE$ is said to be an \textit{$k$-ensemble} if
$$
\dist(\varrho_1, \varrho_2) > 2^k \quad \text{$\varrho_1, \varrho_2 \in \ensE$, $\varrho_1\neq \varrho_2$}.
$$

Denote also $A_k(\varrho):=\abs{\SqrsSetOf{k}{\varrho}}$ for $k\ge0$ and set $A_{-1}(\varrho) := A_0(\varrho) = \abs{\supp\varrho}$.

\subsubsection{The Basic Lemma}
\begin{lemma} \label{lemma21}
Let $k\geq0$ and let $\ensE$ be an ensemble.
There exists a positive absolute constant $C_1$, a finite family of $k$-ensembles $\mathscr{F}'$, with $\ensE\to\ensE'$ for every $\ensE'\in\mathscr{F}'$, positive $(c_{\ensE'})$, $\ensE'\in\mathscr{F}'$, summing to $1$, and real $(K'(\varrho))$, $\varrho\in\ensE'\in\mathscr{F}'$, such that for every $\psi:\Lambda\to\mathbb{R}$,
\begin{equation} \label{label:eq211}
\prod_{\varrho \in \ensE} (1+K(\varrho)\cos\psiOf{\varrho}) = \sum_{\ensE'\in\mathscr{F}'} c_{\ensE'} \prod_{\varrho \in \ensE'} (1+K^{\prime}(\varrho)\cos\psiOf{\varrho}).
\end{equation}
For every $\ensE'\in\mathscr{F}'$, the following properties are satisfied for each $\varrho \in \ensE'$:
\begin{enumerate}
\item \label{lemma22:propertyB}
For any two distinct densities $\varrho_1, \varrho_2 \subset \varrho$, compatible with $\ensE$,
$$
\dist(\varrho_1, \varrho_2) \leq 2^k.
$$
\item \label{lemma22:propertyC}
Let $\varepsilon(\varrho, \varrho')\in\{-1,0,1\}$ be such that $\varrho = \sum_{\varrho' \in \ensE} \varepsilon(\varrho, \varrho')\varrho' $. Then
\begin{equation*}
\abs{K^{\prime}(\varrho)} \leq
3^{\abs{\{ \varrho{\text{\tiny{*}}} \in \ensE: \dist(\varrho, \varrho{\text{\tiny{*}}}) \leq 2^k\}}}
\prod_{\varrho' \in \ensE} \abs{K(\varrho')}^{\abs{\varepsilon(\varrho, \varrho')}}.
\end{equation*}
Moreover, if $\ensE$ is a $(k-1)$-ensemble, then
\begin{equation} \label{label:eq212}
\abs{K^{\prime}(\varrho)} \leq e^{C_1 A_{k-1}(\varrho)} \prod_{\varrho' \in \ensE} \abs{K(\varrho')}^{\abs{\varepsilon(\varrho, \varrho')}}.
\end{equation}
\end{enumerate}
\end{lemma}

\begin{proof}[Proof of Lemma \ref{lemma21}]
The lemma follows from iterated application of the trigonometric identity
\begin{multline} \label{label:eq213}
(1+K_1\cos\psiOf{\varrho_1}) (1+K_2\cos\psiOf{\varrho_2}) =
\frac{1}{3}(1+3K_1\cos\psiOf{\varrho_1}) + \frac{1}{3} (1+3K_2\cos\psiOf{\varrho_2}) \\
+ \frac{1}{6} (1+3K_1 K_2\cos\psiOf{\varrho_1 - \varrho_2})
+ \frac{1}{6} (1+3K_1 K_2\cos\psiOf{\varrho_1 + \varrho_2}),
\end{multline}
which expresses the left-hand side as a convex combination of expressions of the form $1+K\cos\psiOf{\varrho}$.
Note that all densities on the right-hand side of \eqref{label:eq213} are compatible with $\{\varrho_1, \varrho_2\}$.

We start by applying \eqref{label:eq213} to two arbitrary factors on the left-hand side of \eqref{label:eq211} corresponding to two densities $\varrho_1, \varrho_2 \in \ensE$, for which
\begin{equation}\label{label:eq214}
\dist(\varrho_1, \varrho_2) \leq 2^{k}
\end{equation}
holds.
(If there are no such factors, the lemma holds trivially).
The right-hand side of identity \eqref{label:eq213} is the inserted on the left-hand side of \eqref{label:eq211}, replacing these two factors by a sum of four terms, and by expanding, we obtain a sum of four products.
If one of the resulting products contains two factors corresponding to two charge densities, $\varrho_1^{\prime}, \varrho_2^{\prime}$, satisfying \eqref{label:eq214}, we apply identity \eqref{label:eq213} again and expand the resulting expression into a sum of products.
We repeat this operation until we obtain a sum over products indexed by ensembles $\ensE'$ with the property that, for arbitrary $\varrho_1, \varrho_2 \in \ensE'$, $\dist(\varrho_1, \varrho_2) > 2^{k}$.
Note that after every application of \eqref{label:eq213}, every product still corresponds to an ensemble, i.e. all the densities in the product are with disjoint supports.

Clearly, the process yields $c_{\ensE'} > 0$, $\sum_\gamma c_{\ensE'} = 1$, and $\ensE \rightarrow \ensE'$ for each $\ensE'$. Property \eqref{lemma22:propertyB} follows directly from \eqref{label:eq214}. We are left with proving property \eqref{lemma22:propertyC}.

At intermediate stages of the above procedure, we have
$$
\prod_{\varrho \in \ensE} (1+K(\varrho)\cos\psiOf{\varrho}) = \sum_{\ensI} c_{\ensI} \prod_{\varrho \in \ensI} (1+K_{\ensI}(\varrho)\cos\psiOf{\varrho}),
$$
with $\ensE \rightarrow \ensI$.
Let $\varrho\in\ensE'$, $\ensE'\in\mathscr{F}'$, and consider an intermediate ensemble $\ensI$ such that $\varrho$ is compatible with $\ensI$ (if $\varrho$ is not compatible with $\ensI$, further operations on the factors corresponding to densities in $\ensI$ can never produce $\varrho$). We consider all possible applications of identity \eqref{label:eq213} to pairs $\{ \varrho_1, \varrho_2\} \subset \ensI$:
\begin{enumerate}[label=\roman*]
  \item \label{proof22:case1}
  $\pm\varrho_1 \subset \varrho, \varrho_2 \cap \varrho=\emptyset$, or
  \item \label{proof22:case2}
  $\pm\varrho_2 \subset \varrho, \varrho_1 \cap \varrho=\emptyset$, or
  \item \label{proof22:case3}
  $\pm\varrho_1 \subset \varrho, \pm\varrho_2 \subset \varrho$, or
  \item \label{proof22:case4}
  $\varrho_1 \cap \varrho=\emptyset, \varrho_2 \cap \varrho=\emptyset$.
\end{enumerate}

Clearly, in case \eqref{proof22:case4} the application of \eqref{label:eq213} have no effect on $K^{\tp}(\varrho)$.

In case \eqref{proof22:case1}, the term on the right-hand side of identity \eqref{label:eq213}, for which $\varrho$ is still compatible with its ensemble,
is the one in which $\varrho_2$ is eliminated
($\varrho$ is not compatible with the ensemble $\ensI_2$, of the term for which $\varrho_1$ is eliminated, as $(\supp\varrho_1) \cap (\cup_{\varrho \in \ensI_1} \supp\varrho) = \emptyset$, and $\varrho$ is also not compatible with the ensembles $\ensI_{\pm}$, for which $\varrho_1\pm\varrho_2 \in \ensI_{\pm}$, since $\varrho_1\pm\varrho_2$ cannot be separated any more).
Thus, the factor $3K_1$ will be a factor of $K^{\prime}(\varrho)$.
Note that in every application of \eqref{label:eq213} of type \eqref{proof22:case1}, $\dist(\varrho, \varrho^{\prime}) \leq 2^{k}$ for some $\varrho' \subset \varrho_2, \varrho' \in \ensE$. Therefore, the total number of applications of identity \eqref{label:eq213}, of type \eqref{proof22:case1}, which affects $K^{\prime}(\varrho)$, is at most
$$
\abs{\{\varrho' \in \ensE: \varrho \cap \varrho^{\prime} = \emptyset, \dist(\varrho, \varrho^{\prime}) \leq 2^{k} \}}.
$$
Case \eqref{proof22:case2} is the same as case \eqref{proof22:case1}, with $\varrho_1$ and $\varrho_2$ interchanged.

In case \eqref{proof22:case3}, the term on the right-hand side of identity \eqref{label:eq213}, for which $\varrho$ is still compatible with its ensemble,
is either the one with a term of the density $\varrho_1 + \varrho_2$, or the one with a term of the density $\varrho_1 - \varrho_2$.
In any case, the factor $3K_1 K_2$ will be a factor of $K^{\prime}(\varrho)$.
The number of such possible applications of \eqref{label:eq213} is exactly $\abs{\{\varrho'\in\ensE: \varrho' \subset \varrho\}}$.

From the discussion of these cases, we now conclude that
$$
\abs{K^{\prime}(\varrho)} \leq 3^{n_{\ensE}(\varrho)} \prod_{\varrho' \in \ensE} \abs{K(\varrho')}^{\abs{\varepsilon(\varrho, \varrho')}},
$$
where
$$
n_{\ensE}(\varrho):= \abs{\{ \varrho^\prime \in \ensE: \dist(\varrho, \varrho^{\prime}) \leq 2^{k}\}}.
$$
To finish the proof, we show that if $\ensE$ is a $(k-1)$-ensemble then $n_{\ensE}(\varrho) \leq C_1 A_{k-1}(\varrho)$, for some positive absolute constant $C_1$. First note that the case of $k=0$ holds trivially.

For $\ell\geq 1$, let $\bar{A}_\ell(\varrho)$ be the minimal number of $2^\ell \times 2^\ell$ squares needed to cover $\{ j \in \Lambda: \dist(j, \supp\varrho) \leq 2^{\ell+1} \}$. Then $\bar{A}_\ell(\varrho) \leq 25 A_\ell(\varrho)$.

Note that for any two distinct densities $\varrho_1, \varrho_2 \in \ensE$, it holds that $\dist(\varrho_1, \varrho_2) > 2^{k-1}$, and therefore, any $2^{k-1} \times 2^{k-1}$ square cannot intersect more than four different charges in $\ensE$ (if more than four densities intersect such a square, divide the square into four $2^{k-2} \times 2^{k-2}$ squares, and so at least one of these parts intersect two different densities. Therefore the distance between these two densities is less or equal to $2^{k-1}$, which is a contradiction). Therefore we obtain
\begin{equation*}
n_{\ensE}(\varrho) \leq 4\bar{A}_{k-1}(\varrho) \leq 100 A_{k-1}(\varrho). \qedhere
\end{equation*}

\end{proof}

We remark that the final $\mathscr{F}'$ of Lemma~\ref{lemma21} may depend on the order in which we apply \eqref{label:eq213}.
\subsubsection{Proof of Theorem \ref{thm:21}}

Let $C(N_j) = \sum_{q=1}^{N_j} e^{-q^2}$. We write
$$
\lambda_j(\psi) = \sum_{q=1}^{N_j} \frac{e^{-q^2}}{C(N_j)}\left(1 + 2\cdot C(N_j)e^{q^2}\hat{\lambda_j}(q)\cos(q\psi)\right),
$$
and hence,
\begin{equation} \label{label:eq52}
\prod_{j \in \Lambda} \lambda_j(\psi(j)) = \sum_{\vec{q}} \xi(\vec{q}) \prod_{j \in \Lambda} \big[ 1 + z_j(q_j)\cos(q_j\psi(j)) \big],
\end{equation}
with
\begin{equation} \label{eq:z21Def}
\vec{q} = (q_j)_{j \in \Lambda}\in\{1,2,\ldots\}^\Lambda, \quad \xi(\vec{q}) = \prod_{j\in\Lambda} \frac{e^{-q_j^2}}{C(N_j)}, \quad z_j(q_j) = 2C(N_j)e^{q_j^2}\hat{\lambda_j}(q_j).
\end{equation}
Note that $\sum_{\vec{q}} \xi(\vec{q}) = 1$.

Therefore, to prove Theorem \ref{thm:21}, it is enough to show that there exists a positive absolute constant $\DTwo$, and a family $\mathscr{F}_{\vec{q}}$ of ensembles, satisfying properties \ref{enum21:a}-\ref{enum21:d} of Theorem \ref{thm:21}, such that
\begin{equation} \label{label:eq21}
\prod_{j \in \Lambda} \big[ 1 + z(q_j)\cos(q_j\psi(j)) \big] =
\sum_{\ens \in \mathscr{F}_{\vec{q}}} c_\ens^{\prime} \prod_{\varrho \in \ens} [1+K(\varrho)\cos\psiOf{\varrho}],
\end{equation}
where $c_\ens^{\prime} > 0$, $\ens\in\mathscr{F}_{\vec{q}}$ and $\sum_{\ens\in\mathscr{F}_{\vec{q}}} c_\ens^{\prime} = 1$.

To prove this, we apply Lemma~\ref{lemma21} iteratively, to an appropriately chosen $k$-ensembles, $\ensE_k$. Roughly speaking, we shall apply the lemma to ensembles which do not satisfy the assertions of Theorem \ref{thm:21}.
We define the iterative process by induction. Let $\ensQ_{-1} := \{\varrho^j\}_{j \in \Lambda}$ where $\varrho^j:=q_j\cdot\delta_j$, and let $\ensN_{-1}:=\emptyset$. Note that $\ensQ_{-1}$ is a $(-1)$-ensemble.
Given $\ensQ_{k-1}$, $\ensN_{k-1}$, $k\geq0$, define $\ensQ_{k}$ to be a $k$-ensemble in the following way.
First, if $\abs{\ensQ_{k-1} \setminus \ensN_{k-1}}\leq 1$, we set $\ensQ_{k} := \ensQ_{k-1}$, $\ensN_{k}:=\ensN_{k-1}$.
Otherwise, apply Lemma~\ref{lemma21} on $\ensE=\ensQ_{k-1}\setminus\ensN_{k-1}$, which is a $(k-1)$-ensemble, to obtain
\begin{align*}
\prod_{\varrho \in \ensQ_{k-1}} (1+K(\varrho)\cos\psiOf{\varrho})
&=
\prod_{\varrho \in \ensN_{k-1}} (1+K(\varrho)\cos\psiOf{\varrho})
\prod_{\varrho \in \ensQ_{k-1}\setminus\ensN_{k-1}} (1+K(\varrho)\cos\psiOf{\varrho}) \\
&=
\prod_{\varrho \in \ensN_{k-1}} (1+K(\varrho)\cos\psiOf{\varrho})
\sum_{\ensE'\in\mathscr{F}'} c_{\ensE'} \prod_{\varrho \in \ensE'} (1+K^{\prime}(\varrho)\cos\psiOf{\varrho}) \\
&=
\sum_{\ensE'\in\mathscr{F}'} c_{\ensE'} \prod_{\varrho \in \ensE'\cup\ensN_{k-1}} (1+K^{\prime}(\varrho)\cos\psiOf{\varrho}).
\end{align*}
For each $\ensE'\in\mathscr{F}'$ we set $\ensQ_k=\ensQ_{k,\ensE'}:= \ensE' \cup \ensN_{k-1}$, and continue the iterations with each $\ensQ_k = \ensQ_{k,\ensE'}$ separately.
Given $\ensQ_k$, we define $\ensN_k$ by first setting $\ensN_k = \ensN_{k-1}$, and then going over the neutral densities $\varrho \in \ensQ_k \setminus \ensN_{k-1}$, sorted in an ascending order of $d(\varrho)$, and setting $\ensN_k=\ensN_k\cup\{\varrho\}$ if

\begin{equation} \label{label:eq217}
\dist(\varrho_1, \varrho) \geq M[\min(d(\varrho_1), d(\varrho))]^\alpha, \quad \text{$\varrho_1\in\ensQ_k$, $Q(\varrho_1)=0$},
\end{equation}
and
\begin{equation} \label{label:eq218}
\dist(\varrho_1, \varrho) \geq Md(\varrho)^\alpha, \quad \varrho_1 \in \ensQ_k \setminus \ensN_k
\end{equation}
($\ensN_k$ might change during this iteration, and we assume that in \eqref{label:eq218} it is the most updated one).

\begin{remark}
One may show that $\ensN_{k}$ is the maximal (with respect to inclusion) ensemble of neutral densities with $\ensN_{k-1}\subset\ensN_{k}$, such that \eqref{label:eq217} and \eqref{label:eq218} hold for every $\varrho\in\ensN_{k}$.
\end{remark}

Since $\Lambda$ is a finite graph, for $k$ large enough every $\ensQ_k\setminus \ensN_{k}$ is either empty or is an ensemble consisting of a single non-neutral density.
We then set $\mathscr{F}_{\vec{q}}$ to be the collection of the ensembles $\ensQ_k$.
It is straightforward to see that \eqref{label:eq21} holds with $c_\ens^{\prime} > 0$, $\ens\in\mathscr{F}_{\vec{q}}$ and $\sum_{\ens\in\mathscr{F}_{\vec{q}}} c_\ens^{\prime} = 1$, and that properties \ref{enum21:a} and \ref{enum21:b} of Theorem \ref{thm:21} are satisfied by each $\ens\in\mathscr{F}_{\vec{q}}$.

We continue to verify properties \ref{enum21:c} and \ref{enum21:d} of Theorem \ref{thm:21}.

\begin{claim} \label{lemmaHelp21}
Let $k\geq0$ and let $\varrho \in \ensQ_k$ be a neutral density with $Md(\varrho)^\alpha \leq 2^k$. Then $\varrho \in \ensN_k$.
\end{claim}
\begin{proof}
If $\varrho \in \ensN_{k-1}$ the claim is trivial.
Otherwise, note that \eqref{label:eq217} and \eqref{label:eq218} hold for every $\varrho_1 \in \ensQ_k\setminus \ensN_{k-1}$, since $\ensQ_k\setminus \ensN_{k-1}$ is a $k$-ensemble and thus
$$
\dist(\varrho, \varrho_1) > 2^k \geq Md(\varrho)^\alpha \geq M[\min(d(\varrho_1), d(\varrho))]^\alpha.
$$
Also, \eqref{label:eq217} holds trivially for $\varrho_1 \in \ensN_{k-1}$, and therefore $\varrho$ must be in $\ensN_k$.
\end{proof}

In order to verify part \ref{enum21:c} of Theorem \ref{thm:21}, let $\varrho$ be a neutral density in an ensemble $\ens$, and suppose $\varrho_1, \varrho_2$ are two distinct densities with disjoint supports, satisfying $\varrho = \varrho_1 + \varrho_2$, and
$$
\dist(\varrho_1, \varrho_2) \geq 2M[\min(d(\varrho_1), d(\varrho_2))]^\alpha := R.
$$
Without loss of generality, we assume $d(\varrho_1) \leq d(\varrho_2)$, and thus $R=2Md(\varrho_1)^\alpha$.
Let $\ensQ_k$ be the first ensemble (i.e., with with minimal $k$) and $\varrho{\text{\tiny{*}}}\in\ensQ_k$ such that $\ensQ_k\to\ens$, and $\supp\varrho_1 \cap \supp\varrho{\text{\tiny{*}}}\neq\emptyset$, $\supp\varrho_2 \cap \supp\varrho{\text{\tiny{*}}} \neq \emptyset$.
Then there exist $\varrho_\mu, \varrho_\nu \in \ensQ_{k-1}$ with $\varrho_\mu, \varrho_\nu \subset \varrho*$ and $\varrho_\mu\subset\varrho_1$, $\varrho_\nu\subset\varrho_2$, and thus $\dist(\varrho_\mu , \varrho_\nu) \geq R$. Combining this with property \eqref{lemma22:propertyB} of Lemma~\ref{lemma21} yields $R \leq 2^{k}$.
Since $R = 2Md(\varrho_1)^\alpha \geq 2Md(\varrho_\mu)^\alpha$, it holds that $2^{k-1} \geq Md(\varrho_\mu)^\alpha$, and thus $\varrho_\mu$ must be non-neutral (otherwise $\varrho_\mu \in \ensN_{k-1}$ by Claim \ref{lemmaHelp21}).
Now, since $d(\varrho_1) <  Md(\varrho_1)^\alpha \leq 2^{k-1}$ it must be that $\varrho_1 = \varrho_\mu$, as there is no $\varrho_\lambda\in\ensQ_{k-1}$ with $\dist(\varrho_\mu, \varrho_\lambda) \leq 2^{k-1}$. Therefore, $\varrho_1$ is non-neutral, and by $Q(\varrho)=0$, $\varrho_2$ is also non-neutral.

Now we turn to the verification of the bound \eqref{label:eq26}. We choose some $\ens\in\mathscr{F}$. Let $\varrho$ be a neutral density in $\ens$. We start by claiming that $\varrho \in \ensN_{n(\varrho)}$, where $n(\varrho) = \lceil \log_2(Md(\varrho)^\alpha) \rceil$ as in \eqref{label:eqNvarrho}.
First note that $\varrho \in \ensQ_{n(\varrho)}$, since by definition $d(\varrho) \leq 2^{n(\varrho)}$, and so there are no distinct $\varrho_1, \varrho_2 \in \ensQ_{n(\varrho)}\setminus\ensN_{n(\varrho)}$, with $\varrho_1, \varrho_2 \subset \varrho$ (as $\dist(\varrho_1, \varrho_2) > 2^{n(\varrho)} \geq d(\varrho)$), since $\ensQ_{n(\varrho)}\setminus\ensN_{n(\varrho)}$ is an $n(\varrho)$-ensemble.
Claim \ref{lemmaHelp21} yields $\varrho \in \ensN_{n(\varrho)}$.
Denote by $m \leq n(\varrho)$ the minimal $m$ such that $\varrho \in \ensN_m$.
By \eqref{label:eq212},
$$
\abs{K(\varrho)} \leq e^{C_1 A_{m-1}(\varrho)} \prod_{\varrho_\gamma} \abs{K(\varrho_\gamma)},
$$
where $\varrho = \sum \varepsilon(\varrho, \varrho_\gamma)\varrho_\gamma$ with $\varepsilon(\varrho, \varrho_\gamma)\in{\pm1}$, and all densities $\varrho_\gamma$ belong to some $(m-1)$-ensemble. If $m \geq 2$, we apply \eqref{label:eq212} again, which yields
\begin{equation} \label{label:eq221}
\abs{K(\varrho)} \leq \exp \Big\{ C_1 \Big[ A_{m-1}(\varrho) + \sum_{\varrho_\gamma \subset \varrho} A_{m-2}(\varrho_\gamma)\Big] \Big\} \prod_{\varrho_\delta} \abs{K(\varrho_\delta)},
\end{equation}
for densities $\varrho_\delta$ in some $(m-2)$-ensemble. Note that
\begin{equation} \label{label:eq222}
\sum_{\varrho_\gamma \subset \varrho} A_{m-2}(\varrho_\gamma) \leq 4 A_{m-2}(\varrho),
\end{equation}
since every square $s$ of size $2^{m-2} \times 2^{m-2}$ might intersect with at most four densities $\varrho_\gamma$, and since a cover of $\varrho$ with squares of size $2^{m-2} \times 2^{m-2}$, is also a cover of $(\varrho_\gamma)$ for all $\gamma$.

Applying \eqref{label:eq221} and \eqref{label:eq222} iteratively shows that
$$
\abs{K(\varrho)} \leq e^{ \DTwo A(\varrho)} \prod_{j\in\supp\varrho} \abs{z(\abs{\varrho_j})},
$$
and using \eqref{eq:z21Def} and the fact that $2C(N_j)<1$ verifies \eqref{label:eq26}.

\subsection{Spin Wave Construction}\label{sec:chap4}

In this section we prove Proposition~\ref{prop:aExistance}, constructing the spin waves $a_{\varrho}=a_{\varrho, \ens_0}$, $\varrho\in\ens_0$, where $\ens_0$ is a sub-ensemble of neutral densities, of an ensemble $\ens$ satisfying properties \ref{enum21:a}-\ref{enum21:c} of Theorem~\ref{thm:21}.
We fix such ensemble and $\varrho{\text{\tiny{*}}}\in\ens_0$ for the rest of the section. The spin wave $a_{\varrho{\text{\tiny{*}}}}$ is constructed as a sum of spin waves which are defined later in the section. In this section all densities are neutral, and recall that for a neutral density it holds that $d(\varrho)\geq1$.

\subsubsection{Initial Spin Wave}

The next lemma establishes the basic properties of $a_{0, \varrho{\text{\tiny{*}}}}$, an intermediate spin wave which is used to construct $a_{\varrho{\text{\tiny{*}}}}$, which are very similar to the properties of Proposition~\ref{prop:aExistance} (recall the definitions of $D^{+}(\varrho')$, $\varrho\in\ens_0$ and $\ens_0(\varrho{\text{\tiny{*}}})$, defined in \eqref{eq:D_plus_def} and \eqref{eq:NRhoDef} respectively, which appear in this context).

\begin{lemma} \label{lemma:a_0_construction}
There exists $a_{0,\varrho{\text{\tiny{*}}}}:\Lambda\to\mathbb{R}$ such that:
\begin{enumerate}
\item \label{enum0:a0effect} $\supp a_{0,\varrho{\text{\tiny{*}}}} \subset \{j\in\Lambda: \dist(j,\varrho{\text{\tiny{*}}})\leq 1 \}$,
\item \label{enum2:a0effect}
$\supp(a_{0, \varrho{\text{\tiny{*}}}}), \supp(\Delta_\Lambda a_{0, \varrho{\text{\tiny{*}}}}) \subset D(\varrho{\text{\tiny{*}}})$.
\item \label{enum1:a0effect}
For every $\varrho^\prime\in\ens_0(\varrho{\text{\tiny{*}}})$, $a_{0,\varrho{\text{\tiny{*}}}}$ is constant on $D^{+}(\varrho^\prime)$.
\item \label{enum3:a0effect}
$
E_\beta(a_{0, \varrho{\text{\tiny{*}}}}, \varrho{\text{\tiny{*}}}) =
\langle \varrho{\text{\tiny{*}}}, a_{0, \varrho{\text{\tiny{*}}}}\rangle - \frac{\beta}{2}\gradNormUp{a_{0, \varrho{\text{\tiny{*}}}}} \geq
\frac{1}{16\beta}\norm{\varrho{\text{\tiny{*}}}}_2^2.
$
\end{enumerate}
\end{lemma}

\begin{proof}
We start by defining $a_{0, \varrho{\text{\tiny{*}}}}$. Recall that $\Lambda$ is bipartite. Let $(\Omega_1, \Omega_2)$ be a bipartition of $\Lambda$, chosen without loss generality so that
\begin{equation} \label{eqOmega1Cond}
\sum_{j \in \Omega_1}\varrho{\text{\tiny{*}}}(j)^2 \geq \frac{1}{2} \sum_{j \in \Lambda}\varrho{\text{\tiny{*}}}(j)^2
\end{equation}
In the case that $d(\varrho{\text{\tiny{*}}})=1$, we assume also that $\Omega_1$ contains the center of $D(\varrho{\text{\tiny{*}}})$, which is safe since $\varrho{\text{\tiny{*}}}$ is neutral (recall the definition of $D(\varrho)$ for a density $\varrho$ in \eqref{eq:DRhoDef}).
Let $\varrho_\ell^*(j) = \varrho{\text{\tiny{*}}}(j)$ for $j \in \Omega_l$, and $\varrho_\ell^*(j)=0$ otherwise, $\ell=1,2$. Clearly $\varrho_1^* + \varrho_2^* = \varrho{\text{\tiny{*}}}$.
Note that \eqref{eqOmega1Cond} implies $\norm{\varrho_1^*}_2^2 \geq \frac{1}{2} \norm{\varrho{\text{\tiny{*}}}}_2^2$.

We denote by $d_j$ the graph degree of $j\in\Lambda$, and define $a_{0, \varrho{\text{\tiny{*}}}}:\Lambda\to\mathbb{R}$ by
\begin{equation} \label{eq:a0Def}
a_{0, \varrho{\text{\tiny{*}}}}(j):= \frac{\varrho_1^*(j)}{d_j\beta}.
\end{equation}

Next, we verify the properties of the lemma. Property \ref{enum0:a0effect} holds by definition.
Property \ref{enum2:a0effect} follows by property \ref{enum0:a0effect} and $\supp(\Delta_\Lambda a_{0, \varrho{\text{\tiny{*}}}}) \subset \{j:\dist(j,\varrho_1^*)\leq 1 \}$, where for the case $d(\varrho{\text{\tiny{*}}}) = 1$ we used the assumption that $\Omega_1$ contains the center of $D(\varrho{\text{\tiny{*}}})$.
For property \ref{enum1:a0effect}, for every $\varrho'\in\ens_0(\varrho{\text{\tiny{*}}})$ it follows by property \ref{enum21:b} of Theorem~\ref{thm:21} that
$$
\dist(D(\varrho'), \varrho{\text{\tiny{*}}}) \geq \dist(\varrho', \varrho{\text{\tiny{*}}}) - 4d(\varrho') \geq \frac{M}{4}d(\varrho') > 2,
$$
which implies that $a_{0,\varrho{\text{\tiny{*}}}}$ is zero on $D(\varrho')$, $\varrho'\in\ens_0(\varrho{\text{\tiny{*}}})$.

Part \ref{enum3:a0effect} follows by the observation that the gradient is $\varrho_1^*(j)/(d_j\beta)$ for every edge touching $j$ and $0$ otherwise, and thus by \eqref{eqOmega1Cond}
\begin{equation*}
E_\beta(a_{0, \varrho{\text{\tiny{*}}}}, \varrho{\text{\tiny{*}}})
=
\langle \varrho{\text{\tiny{*}}}, a_{0, \varrho{\text{\tiny{*}}}}\rangle - \frac{\beta}{2}\gradNormUp{a_{0, \varrho{\text{\tiny{*}}}}}
=
\frac{1}{\beta}\sum_{j\in\Lambda} \frac{\varrho_1^*(j)^2}{d_j} -
\frac{1}{2\beta}\sum_{j\in\Lambda} \frac{\varrho_1^*(j)^2}{d_j}
\geq
\frac{1}{16\beta}\norm{\varrho{\text{\tiny{*}}}}_2^2.\qedhere
\end{equation*}

\end{proof}

\subsubsection{Connected Components of $\bigcup_{\varrho \in \ens_0(\varrho{\text{\tiny{*}}})}D^+(\varrho)$}

Recall requirement \ref{enum1:aExistance} of Proposition~\ref{prop:aExistance}, which states that $a_{\varrho{\text{\tiny{*}}}}$ is constant on every $D^{+}(\varrho^\prime)$ for $\varrho^\prime\in\ens_0(\varrho{\text{\tiny{*}}})$, where $\ens_0(\varrho{\text{\tiny{*}}}), D^{+}(\varrho^\prime)$ are defined in \eqref{eq:NRhoDef} and \eqref{eq:D_plus_def} respectively.
It may happen that $D^{+}(\varrho_1)$ and $D^{+}(\varrho_2)$ overlap for some $\varrho_1, \varrho_2 \in \ens_0(\varrho{\text{\tiny{*}}})$, whence the requirement becomes that $a_{\varrho{\text{\tiny{*}}}}$ is constant on the union $D^{+}(\varrho_1)\cup D^{+}(\varrho_2)$.
Therefore, the connected components of the $\bigcup_{\varrho \in \ens_0(\varrho{\text{\tiny{*}}})}D^+(\varrho)$ are important. When we say a connected component, we think of $\bigcup_{\varrho \in \ens_0(\varrho{\text{\tiny{*}}})}D^+(\varrho)$ as a sub-graph of $\Lambda$.
The next lemma provides bounds we shall use later.

Denote the external vertex boundary of a set $E\subset\Lambda$ by
$$
\partial^{\ext} E := \{j\in\Lambda: \dist(j,E)=1\}.
$$
For each $E\subset\Lambda$ denote its diameter by
$$
d(E) := \max_{j,\ell \in E} \dist(j,\ell).
$$
Note that for $E$, a connected component of $\bigcup_{\varrho \in \ens_0(\varrho{\text{\tiny{*}}})}D^+(\varrho)$, it holds that $d(E)\geq4$, since every $\varrho\in\ens_0(\varrho{\text{\tiny{*}}})$ is neutral, and thus has $d(\varrho)\geq1$.

\begin{lemma} \label{lemma:appendixE}
For every connected component $E$ of $\ \bigcup_{\varrho \in \ens_0(\varrho{\text{\tiny{*}}})} D^+(\varrho)$ it holds that
\begin{enumerate}
\item \label{item:appendixE_1} For any connected component $E'$ of $\ \bigcup_{\varrho \in \ens_0(\varrho{\text{\tiny{*}}})} D^+(\varrho)$ with $E\neq E'$,
$$
\frac{M}{25}\left[\min(d(E), d(E'))\right]^\alpha \leq \dist(E, E') + d(E)+d(E'),
$$
\item \label{item:appendixE_2}$
\frac{M}{128}d(E)^\alpha \leq \dist(E, \varrho{\text{\tiny{*}}}),
$
\item \label{item:appendixE_3} $\abs{\partial^{\ext} E} \leq 64 \cdot d(E)$.
\end{enumerate}
\end{lemma}

\begin{proof}
We prove the following claim, which is a stronger variant of the lemma.
Let $\ensI\subseteq\ens_0(\varrho{\text{\tiny{*}}})$ be a sub-ensemble.
For every connected component (c.c.) $E$ of $\bigcup_{\varrho \in \ensI} D^+(\varrho)$, writing $\ensS_E:=\{\varrho\in\ensI: D^+(\varrho)\subseteq E\}$, there exists a unique density $\varrho_E$ for which $d(\varrho_E) = \max\{d(\varrho):\varrho\in \ensS_E\}$. Also,
\begin{equation} \label{eq:diamEBound}
d(E) \leq 5 d(\varrho_E),
\end{equation}
and properties \ref{item:appendixE_2} and \ref{item:appendixE_3} of the lemma hold for such $E$, and property \ref{item:appendixE_1} holds for $E$ and any other connected component $E'$ of $\bigcup_{\varrho \in \ensI} D^+(\varrho)$, $E'\neq E$.

We prove this claim by induction on the size of $\ensI \subset \ens_0(\varrho{\text{\tiny{*}}})$. Note that the claim trivially holds for $\ensI=\emptyset$. For the induction step, let $\ensI \subset \ens_0(\varrho{\text{\tiny{*}}})$ with $\abs{\ensI} > 0$, let $E$ be a connected component of $\bigcup_{\varrho \in \ensI} D^+(\varrho)$, let $\varrho_{\mathop{m}}$ be some arbitrary density in $\ensS_E$ with maximal diameter and let
$$
\mathscr{M} = \left\{ \varrho \in \ensS_E: d(\varrho) \geq \frac{1}{10}d(\varrho_{\mathop{m}}) \right\}.
$$

Assume that there are two distinct $\varrho_1, \varrho_2 \in \mathscr{M}$, then by \eqref{eq:alphaMBounds} and property \ref{enum21:b} of Theorem~\ref{thm:21}, $D^+(\varrho_1),D^+(\varrho_2)$ are disjoint, since otherwise
$$
M\Big( \frac{d(\varrho_m)}{10} \Big)^\alpha \leq \dist(\varrho_1, \varrho_2) \leq d(D^+(\varrho_1)) + d(D^+(\varrho_2))\leq 4d(\varrho_1) + 4d(\varrho_2) \leq 8d(\varrho_m).
$$
Since $E$ is connected, there exist two distinct $\varrho_1, \varrho_2 \in \mathscr{M}$ which are both intersecting some $E'$, a connected components of $\bigcup_{\varrho \in \ensS_E \setminus \mathscr{M}} D^+(\varrho)$, which is impossible by \eqref{eq:alphaMBounds}, property \ref{enum21:b} of Theorem~\ref{thm:21} and the induction hypothesis applied to $\ensS_E\setminus\mathscr{M}$, as
$$
M\Big( \frac{d(\varrho_m)}{10} \Big)^\alpha \leq \dist(\varrho_1, \varrho_2) \leq d(D^+(\varrho_1)) + d(D^+(\varrho_2)) + d(E')\leq 8d(\varrho_m) + 5d(\varrho_{E'}) \leq 9d(\varrho_m).
$$

Therefore $\abs{\mathscr{M}} = 1$, which implies that there is a unique density $\varrho_E \in \ensS_E$ with maximal diameter. Let $P$ be a path in $E$ satisfying that the distance between its endpoints equals $d(E)$. Let $x,y\in P$ be the first and last points along $P$ which lie in $D^+(\varrho_E)$. Each of the portions of $P$ before $x$ and after $y$ must lie in a connected component of $\bigcup_{\ensS_E\setminus\mathscr{M}} D^+(\varrho)$. Thus, applying the induction hypothesis to the sub-ensemble $\ensS_E\setminus\mathscr{M}$,
\begin{equation*}
  d(E) \le 2\cdot 5\cdot\frac{1}{10} d(\varrho_E) + \dist(x,y)\le d(\varrho_E) + 4d(\varrho_E) \le 5d(\varrho_E),
\end{equation*}
which proves \eqref{eq:diamEBound}.

To prove properties \ref{item:appendixE_1} and \ref{item:appendixE_2}, let $E$, $E'$ be two distinct connected components of $\bigcup_{\varrho \in \ensI} D^+(\varrho)$. By property \ref{enum21:b} of Theorem~\ref{thm:21} and \eqref{eq:diamEBound} we obtain
$$
M\left[\frac{\min(d(E), d(E'))}{5}\right]^\alpha \leq
M\left[\min(d(\varrho_E), d(\varrho_{E'}))\right]^\alpha \leq
\dist(\varrho_{E},\varrho_{E'}) \leq
\dist(E, E') + d(E)+d(E'),
$$
which proves property \ref{item:appendixE_1}. Property \ref{item:appendixE_2} follows similarly by
$$
M\left[\frac{d(E)}{10}\right]^\alpha \leq
M\left[\frac{d(\varrho_E)}{2}\right]^\alpha \leq
\dist(\varrho_{E},\varrho{\text{\tiny{*}}}) \leq
\dist(E, \varrho{\text{\tiny{*}}}) + d(E).
$$

To prove property \ref{item:appendixE_3}, note that by the induction hypothesis,
\begin{equation} \label{eq:circle_size}
\abs{\partial^{\ext} E} \leq
\abs{\partial^{\ext} D^+(\varrho_E)} +
\sum_{\substack{
E' \text{ c.c. of } \ensS_E\setminus\mathscr{M} \\
E'\bigcap\partial^{\ext} D^+(\varrho_E)\neq\emptyset}}
\abs{\partial^{\ext} E'} \leq
12 d(\varrho_E) +
16 \cdot \sum_{\substack{
E' \text{ c.c. of } \ensS_E\setminus\mathscr{M} \\
E'\bigcap\partial^{\ext} D^+(\varrho_E)\neq\emptyset}}
d(\varrho_{E'}).
\end{equation}
By property 1 and the induction hypothesis, we immediately obtain that $\dist(E_1, E_2) \geq \frac{M}{32} 2^{\alpha\ell}$ for any two distinct connected components $E_1, E_2$ of $\ensS_E\setminus\mathscr{M}$ with $2^{\ell} \leq d(E_1),d(E_2) \leq 2^{\ell+1}$, $\ell\geq 1$. Therefore, the number of possible $E'$ in this scale, which intersect $\partial^{\ext} D^+(\varrho_E)$, is bounded by
$$
\frac{2\abs{\partial^{\ext} D^+(\varrho_E)}}{\frac{M}{32}2^{\alpha \ell}} \leq \frac{1024}{M} 2^{-\alpha \ell} d(\varrho_E).
$$
Substituting this in \eqref{eq:circle_size} and using \eqref{eq:alphaMBounds} and the induction hypothesis yields
\begin{equation*}
\abs{\partial^{\ext} E} \leq 12d(\varrho_E) + 16d(\varrho_E)\sum_{\ell=1}^\infty \frac{1024}{M} 2^{-\alpha \ell} \cdot 2^{\ell+1} \leq 16 d(\varrho_E) \leq 64 d(E). \qedhere
\end{equation*}
\end{proof}

\newcommand*{\ROne}{2^{k-1}+2^{k-3}}
\newcommand*{\RTwo}{2^{k-1}+2^{k-2}}
\newcommand*{\Rratio}{6/5}
\newcommand*{\RratioFrac}{\frac{6}{5}}

\subsubsection{A Spin Wave for Every Square in $\SqrsSetTagOf{k}{\varrho{\text{\tiny{*}}}}$}

We remind that for each $k\geq 1$ and each $s\in\SqrsSetTagOf{k}{\varrho{\text{\tiny{*}}}}$ the density $s\cap\varrho{\text{\tiny{*}}}$, defined in \eqref{eq:sCapVarrho}, is non-neutral by the definition of $\SqrsSetTagOf{k}{\varrho{\text{\tiny{*}}}}$ and property \ref{enum21:c} of Theorem~\ref{thm:21}.
This fact allow us to construct an intermediate spin wave $a_{s,\varrho{\text{\tiny{*}}}}$, such that $E_\beta(a_{s,\varrho{\text{\tiny{*}}}}, \varrho{\text{\tiny{*}}}) > C/\beta$, for some positive absolute constant $C$.

\begin{prop} \label{prop:a_s_construction}
There exists a positive absolute constant $\DThr$ such that the following holds.
Let $k\geq 1$ be an integer and let $s\in\SqrsSetTagOf{k}{\varrho{\text{\tiny{*}}}}$.
Then there exists $a_{s,\varrho{\text{\tiny{*}}}}:\Lambda\to\mathbb{R}$ such that:
\begin{enumerate}
\item \label{item:a_s_construction1} $\supp a_{s,\varrho{\text{\tiny{*}}}} \subset \{j\in\Lambda: \dist(j,s)\leq 2^{k-1} \}$,
\item \label{item:a_s_construction3} $a_{s,\varrho{\text{\tiny{*}}}}$ is constant on $\{j\in\Lambda: \dist(j,s)\leq \lceil2^{k-3}\rceil \}$,
\item \label{item:a_s_construction4} $a_{s,\varrho{\text{\tiny{*}}}}$ is constant on $D^+(\varrho')$ for every $\varrho'\in\ens_0(\varrho{\text{\tiny{*}}})$,
\item \label{item:a_s_construction5} it holds that
\begin{equation}
E_\beta(a_{s,\varrho{\text{\tiny{*}}}}, \varrho{\text{\tiny{*}}}) = \langle a_{s,\varrho{\text{\tiny{*}}}}, \varrho{\text{\tiny{*}}}\rangle - \frac{\beta}{2}\gradNormUp{a_{s,\varrho{\text{\tiny{*}}}}} \geq
\frac{\DThr}{\beta}.
\end{equation}
\end{enumerate}
\end{prop}

\begin{proof}

{\bf Case 1: $k < 10$.} Fix $k < 10$, and $s\in\SqrsSetTagOf{k}{\varrho{\text{\tiny{*}}}}$.
Define
$$
a_{s,\varrho{\text{\tiny{*}}}}(j) :=
\twopartdef{\frac{q}{2^{k+3}\beta},} {\dist(j,s) \leq 1}
{0,}{\text{otherwise}}
$$
where $q$ is the charge $Q(s \cap \varrho{\text{\tiny{*}}})$ (emphasizing again that $\abs{q} \geq 1$).

It is straightforward to see that properties \ref{item:a_s_construction1}, \ref{item:a_s_construction3} of the proposition are satisfied by this definition. To verify property \ref{item:a_s_construction4}, let $\varrho' \in \ens_0(\varrho{\text{\tiny{*}}})$, then by property \ref{enum21:b} of Theorem~\ref{thm:21}
$$
\dist(D^+(\varrho'), s) \geq \dist(\varrho', \varrho{\text{\tiny{*}}}) - 4d(\varrho') - d(s) \geq M\left( \frac{d(\varrho')}{2} \right)^\alpha - 4d(\varrho') - 2^{k+1} \geq \frac{M}{4} -2^{k+1}  > 2^{k-1},
$$
which shows that $a_{s, \varrho{\text{\tiny{*}}}}$ is zero on every $D^+(\varrho')$, $\varrho'\in\ens_0(\varrho{\text{\tiny{*}}})$. Finally, to verify property \ref{item:a_s_construction5}, note that
\begin{equation*}
E_\beta(a_{s,\varrho{\text{\tiny{*}}}}, \varrho{\text{\tiny{*}}}) = \langle a_{s,\varrho{\text{\tiny{*}}}}, \varrho{\text{\tiny{*}}}\rangle - \frac{\beta}{2}\gradNormUp{a_{s,\varrho{\text{\tiny{*}}}}} =
\frac{q^2}{2^{k+3}\beta} - \frac{\beta}{2} \cdot 2^{k+3} \frac{q^2}{2^{2(k+3)}\beta^2} \geq \frac{1}{2^{k+4}\beta}.
\end{equation*}

{\bf Case 2: $k \geq 10$.}
\newcommand*{\abss}[1]{\abs{#1}_s}
Fix $k\geq 10$ and $s\in\SqrsSetTagOf{k}{\varrho{\text{\tiny{*}}}}$. We define $a_{s,\varrho{\text{\tiny{*}}}}$ using a parameter $\gamma > 0$ which is chosen later to minimize the lower bound on $E_\beta(a_{s, \varrho{\text{\tiny{*}}}}, \varrho{\text{\tiny{*}}})$.

Define a function $\abss{\cdot}:\Lambda\to\mathbb{R}$ that roughly measures the distance between points of $\Lambda$ and the center of $s$ by
\begin{equation*}
\abss{j}:= \dist(j,s) + 2^{k-1}.
\end{equation*}
For notational simplicity, denote
\begin{align*}
R_1 &:= \ROne, \\
R_2 &:= \RTwo = \RratioFrac R_1.
\end{align*}
Define a function $b:\Lambda\to\mathbb{R}$ by
$$
b(j) :=
\threepartdef{\ln(\Rratio),} {\abss{j} \leq R_1}
{\ln[R_2\abss{j}^{-1}],}{R_1 \leq \abss{j} \leq R_2}
{0,}{\abss{j} \geq R_2}.
$$
Let $\mathcal{E}$ be the set of connected components of $\bigcup_{\varrho \in \ens_0(\varrho{\text{\tiny{*}}})}D^+(\varrho)$. We define $a_{s,\varrho{\text{\tiny{*}}}}$ as follows:
\begin{equation} \label{label:eq437}
a_{s,\varrho{\text{\tiny{*}}}}(j) :=
\twopartdef{\gamma b(j),} {j \notin E \text{ for any } E\in\mathcal{E}}
{\gamma b(x_E),}{j \in E\in\mathcal{E}},
\end{equation}
where $x_E\in E$ is chosen so that if $E$ meets $\abss{j} = R_1$ then $\abss{x_E} = R_1$, if $E$ meets $\abss{j} = R_2$ then $\abss{x_E} = R_2$, and otherwise it is arbitrarily chosen point in $E$. Note that $E\in\mathcal{E}$ cannot meet both $\abss{j} \leq R_1$ and $\abss{j} \geq R_2$, since $\dist(E, \varrho{\text{\tiny{*}}})\leq 2R_1$ and thus by \eqref{eq:alphaMBounds} and property \ref{item:appendixE_2} of Lemma~\ref{lemma:appendixE} it follows that
$$
d(E) \leq  \frac{256}{M} R_1 < \frac{R_2-R_1}{2}.
$$

It is straightforward to see that properties \ref{item:a_s_construction1}, \ref{item:a_s_construction3}, \ref{item:a_s_construction4} of the proposition are satisfied by this definition.
We verify property \ref{item:a_s_construction5} next.
First, note that since $a_{s,\varrho{\text{\tiny{*}}}}$ is valued $\gamma\ln(\Rratio)$ on $s\cap\varrho{\text{\tiny{*}}}$, and is valued 0 on $\varrho{\text{\tiny{*}}} - (s\cap\varrho{\text{\tiny{*}}})$ (since by property \ref{enum21:c} of Theorem~\ref{thm:21} $\dist(s,\varrho{\text{\tiny{*}}} - (s\cap\varrho{\text{\tiny{*}}})) \geq 2M2^{\alpha (k+1)} > R_2$), and thus
\begin{equation}\label{eq:a_rho_bound}
\langle a_{s,\varrho{\text{\tiny{*}}}}, \varrho{\text{\tiny{*}}}\rangle = \gamma \ln(\Rratio) Q(s\cap \varrho{\text{\tiny{*}}}).
\end{equation}
In the rest of the proof we show that there exists a positive absolute constant $C$ such that
\begin{equation} \label{eq:grad_a_bound}
\gradNormUp{a_{s,\varrho{\text{\tiny{*}}}}} \leq
\gamma^2 C,
\end{equation}
which yields (using \eqref{eq:a_rho_bound})
$$
E_\beta(a_{s,\varrho{\text{\tiny{*}}}}, \varrho{\text{\tiny{*}}}) \geq \gamma \ln(\Rratio) Q(s\cap \varrho{\text{\tiny{*}}}) - \frac{\beta}{2}\gamma^2 C \ln(2/\upsilon),
$$
and optimization on $\gamma$ (using $\abs{Q(s\cap\varrho{\text{\tiny{*}}})} \geq 1$) verifies property \ref{item:a_s_construction5} of the proposition.

To show \eqref{eq:grad_a_bound}, note that by the definition of $a_{s,\varrho{\text{\tiny{*}}}}$,
\begin{equation} \label{label:eq442}
\gradNormUp{a_{s,\varrho{\text{\tiny{*}}}}} =
\gamma^2
\sum_{
\substack{j\thicksim \ell \\ j,\ell \notin \bigcup_{E\in\mathcal{E}} E}
}(b(j)-b(\ell))^2
+ \gamma^2\sum_{E\in\mathcal{E}}\sum_{j\in\partial^{\ext} E} (b(x_E) - b(j))^2.
\end{equation}

To bound the first sum on the right-hand side of \eqref{label:eq442}, note that there are at most $2^{2(k+1)}$ edges involved in the sum, and for each such edge $\{j,\ell\}\in E(\Lambda)$,
$$
(b(j)-b(\ell))^2 \leq \left(\ln\left( 1 + \frac{1}{\min(\abss{j}, \abss{\ell})}\right)\right)^2 \leq \frac{1}{\min(\abss{j}, \abss{\ell})^2} \leq 2^{-2(k-1)},
$$
and therefore,
\begin{equation} \label{eq:first_a_bound}
\sum_{
\substack{j\thicksim \ell \\ j,\ell \notin \bigcup_{E\in\mathcal{E}} E}
}(b(j)-b(\ell))^2 \leq 2^{2(k+1) - 2(k-1)} = 16.
\end{equation}

We continue with bounding the second sum on the right-hand side of \eqref{label:eq442}. Denote by $\mathcal{E}_0\subseteq\mathcal{E}$ the set of connected components $E\in\mathcal{E}$ such that $R_1 \leq \abss{x_E} \leq R_2$, and note that $\sum_{j\in\partial^{\ext} E} (b(x_E) - b(j))^2 = 0$ for $E\in\mathcal{E}\setminus\mathcal{E}_0$.

Let $E\in\mathcal{E}_0$. We show next that for some absolute constant $C_2$ it holds that
\begin{equation} \label{eq:bSecSumBound}
\sum_{j\in\partial^{\ext} E} (b(x_E) - b(j))^2 \leq C_2 d(E)^3 \abss{x_E}^{-2}.
\end{equation}
Indeed, let $j\in\partial^{\ext}E$.
Using
\begin{equation} \label{eq:dE_xE_rel}
\dist(E,\varrho{\text{\tiny{*}}}) \leq \dist(x_E, s) + 2R_1 \leq 3\abss{x_E},
\end{equation}
and property \ref{item:appendixE_2} of Lemma~\ref{lemma:appendixE}, we obtain that
$$
\big|\abss{j}-\abss{x_E}\big| \leq 2d(E) \leq \frac{256}{M} \dist(E,\varrho{\text{\tiny{*}}}) \leq \frac{1}{2}\abss{x_E},
$$
and by $(\ln(1+t))^2 \leq 10 t^2$ for $\abs{t} < \frac{1}{2}$,
$$
(b(j) - b(x_E))^2 \leq \big( \ln\abss{j} - \ln\abss{x_E} \big) ^ 2
= \left( \ln\left(1 + \frac{\abss{j}-\abss{x_E}}{\abss{x_E}} \right) \right)^2
\leq \frac{(\abss{j}-\abss{x_E})^2}{\abss{x_E}^2} \leq \frac{4d(E)^2}{\abss{x_E}^2}.
$$
Thus, by property \ref{item:appendixE_3} of Lemma~\ref{lemma:appendixE},
\begin{equation*}
\sum_{j\in\partial^{\ext} E} (b(x_E) - b(j))^2 \leq
\abs{\partial^{\ext}E} \frac{4d(E)^2}{\abss{x_E}^2} \leq 256 {d(E)^3}{\abss{x_E}^2},
\end{equation*}
which proves \eqref{eq:bSecSumBound}.

To bound the right-hand side of \eqref{eq:bSecSumBound}, we split $\mathcal{E}_0$ into a disjoint union of $\mathcal{E}_{\ell, r}$, indexed by positive integers $r,\ell$, and defined by
$$
\mathcal{E}_{\ell, r} := \{ E\in\mathcal{E}_0: 2^{\ell}\leq d(E) < 2^{\ell+1}, \quad r c_{\ell}\leq\abss{x_E}<(r+1)c_{\ell}\}, \quad\quad c_{\ell} := \frac{M}{512}2^{\alpha \ell}.
$$
This is a partition of $\mathcal{E}_0$, since $d(E)\geq 4$, and since for every $E\in\mathcal{E}_0$ with $2^{\ell}\leq d(E) < 2^{\ell+1}$ it holds that $\abss{x_E}\geq c_\ell$, by property \ref{item:appendixE_2} of Lemma~\ref{lemma:appendixE} and \eqref{eq:bSecSumBound}.

We continue by bounding $\abs{\mathcal{E}_{\ell, r}}$. Fix $\ell,r\geq1$. By property \ref{item:appendixE_1} of Lemma~\ref{lemma:appendixE} it holds that
$$
\dist(x_{E_1}, x_{E_2}) \geq \dist(E_1, E_2) \geq \frac{M}{25} 2^{\alpha \ell} - 2\cdot 2^{\ell+1} \geq \frac{M}{32} 2^{\alpha \ell} \geq 3 c_\ell,
\quad E_1, E_2 \in \mathcal{E}_{\ell, r},
$$
and therefore the sets $\{j\in\Lambda: \dist(j,x_E) \leq c_\ell \}$, $E\in\mathcal{E}_{\ell, r}$ are disjoint, which implies
\begin{equation} \label{eq:E_0_size_bound}
\abs{\mathcal{E}_{\ell, r}} \leq
\frac
{\abs{\{j: r c_\ell\leq\abss{j}<(r+1)c_\ell\}}}
{\abs{\{j\in\Lambda: \dist(j,x_E) \leq c_\ell \}}} \leq C_3 r,
\end{equation}
for some positive absolute constant $C_3$.
Also, if $r c_\ell > R_2$ or $(r+1)c_\ell <  R_1$ then $\abs{\mathcal{E}_{\ell, r}}=0$.

Thus, by \eqref{eq:alphaMBounds}, \eqref{eq:bSecSumBound} and \eqref{eq:E_0_size_bound},
\begin{multline} \label{eq:SumOnEBound}
\sum_{E\in\mathcal{E}_0}\sum_{j\in\partial^{\ext} E} (b(x_E) - b(j))^2
\leq
C_2 \sum_{E\in\mathcal{E}_0} d(E)^3 \abss{x_E}^{-2}
\leq
C_2 \sum_{\ell=1}^{\infty}
\sum_{
\substack{r\geq 1 \\ \lceil R_1/c_\ell\rceil-1\leq r \leq \lfloor R_2/c_\ell\rfloor}
}
\abs{\mathcal{E}_{\ell, r}} \cdot 2^{3(\ell+1)} \cdot \frac{1}{(r c_\ell)^2} \\
\leq
C_2 C_3 2^3\sum_{\ell=1}^{\infty} 2^{(3-2\alpha)\ell}
\sum_{
\substack{r\geq 1 \\ \lceil R_1/c_\ell\rceil-1\leq r \leq \lfloor R_2/c_\ell\rfloor}
}
\frac{1}{r}
\leq
C_4 \sum_{\ell=1}^{\infty} 2^{(3-2\alpha)\ell} \ln\frac{R_2}{R_1} \leq C_5.
\end{multline}

Combining \eqref{label:eq442}, \eqref{eq:first_a_bound} and \eqref{eq:SumOnEBound} yields \eqref{eq:grad_a_bound}, finishing the proof of Proposition~\ref{prop:a_s_construction}.
\end{proof}


\subsubsection{Proof of Proposition \ref{prop:aExistance} - Spin Wave Existence} \label{sec:proofOfSW}
Let $a_{0, \varrho{\text{\tiny{*}}}}$ be be such that the properties of Lemma~\ref{lemma:a_0_construction} hold, and let $a_{s,\varrho{\text{\tiny{*}}}}$, $s\in\SqrsSetTagOf{k}{\varrho{\text{\tiny{*}}}}$, $k\geq 1$, be such that the properties of Proposition~\ref{prop:a_s_construction} hold, and define
$$
a_{\varrho{\text{\tiny{*}}}} := a_{0, \varrho{\text{\tiny{*}}}} + \sum_{k=1}^\infty \sum_{s\in\SqrsSetTagOf{k}{\varrho{\text{\tiny{*}}}}} a_{s, \varrho{\text{\tiny{*}}}}.
$$
We verify next the properties of Proposition~\ref{prop:aExistance} for $a_{\varrho{\text{\tiny{*}}}}$. Property \ref{enum1:aExistance} is immediate by Lemma~\ref{lemma:a_0_construction} and by Proposition~\ref{prop:a_s_construction}. Properties \ref{enum2:aExistance} and \ref{enum3:aExistance} hold
for $a_{0,\varrho{\text{\tiny{*}}}}$ by Lemma~\ref{lemma:a_0_construction}. They hold also for every $a_{s,\varrho{\text{\tiny{*}}}}$, $s\in\SqrsSetTagOf{k}{\varrho{\text{\tiny{*}}}}$, $k\geq 1$, since given such $s$, by the definition of $\SqrsSetTagOf{k}{\varrho{\text{\tiny{*}}}}$ it holds that
$$
d(\varrho{\text{\tiny{*}}}) \geq \dist(s, \varrho{\text{\tiny{*}}} - \varrho{\text{\tiny{*}}} \cap s) \geq 2M 2^{\alpha(k+1)},
$$
which implies using Proposition~\ref{prop:a_s_construction} that $\supp a_{s,\varrho{\text{\tiny{*}}}} \subset \{j\in\Lambda: \dist(j,s)\leq 2^{k-1} \} \subset D(\varrho{\text{\tiny{*}}})$.

We are thus left with verifying property \ref{enum4:aExistance}. For that we show that for every edge of $\Lambda$, the gradient on this edge is non zero in at most one $a_{t, \varrho{\text{\tiny{*}}}}$, and thus it follows by Lemma~\ref{lemma:a_0_construction} and by Proposition~\ref{prop:a_s_construction} that
\begin{multline} \label{eq:aTotBound}
E_\beta(a_{\varrho{\text{\tiny{*}}}}, \varrho{\text{\tiny{*}}}) =
\langle a_{\varrho{\text{\tiny{*}}}}, \varrho{\text{\tiny{*}}}\rangle - \frac{\beta}{2}\gradNormUp{a_{\varrho{\text{\tiny{*}}}}}
=
\sum_t \langle a_{t,\varrho{\text{\tiny{*}}}}, \varrho{\text{\tiny{*}}}\rangle
-\frac{\beta}{2}\sum_t\gradNormUp{a_{t,\varrho{\text{\tiny{*}}}}}
\\
=
E_\beta(a_{0,\varrho{\text{\tiny{*}}}}, \varrho{\text{\tiny{*}}}) + \sum_{k=1}^\infty\sum_{s\in\SqrsSetTagOf{k}{\varrho{\text{\tiny{*}}}}} E_\beta(a_{s,\varrho{\text{\tiny{*}}}}, \varrho{\text{\tiny{*}}})
\geq
\frac{1}{\beta}\left( \frac{1}{16} \norm{\varrho{\text{\tiny{*}}}}_2^2 + \DThr \sum_{k=1}^\infty \abs{\SqrsSetTagOf{k}{\varrho{\text{\tiny{*}}}}} \right).
\end{multline}

First note that for any $s\in\SqrsSetTagOf{k}{\varrho{\text{\tiny{*}}}}$, $k\geq1$, all of the non-zero gradients of $a_{0, \varrho{\text{\tiny{*}}}}$ are on edges inside the square $\{j:\dist(j,s)\leq 1 \}$, in which the gradients of $a_{s, \varrho{\text{\tiny{*}}}}$ are zero by Proposition~\ref{prop:a_s_construction}.

Second, let $s\in\SqrsSetTagOf{k}{\varrho{\text{\tiny{*}}}}$, $k\geq1$. By the definition of $\SqrsSetTagOf{k}{\varrho{\text{\tiny{*}}}}$ the only gradients which might interfere with the non-zero gradients of $a_{s, \varrho{\text{\tiny{*}}}}$, are the non-zero gradients of squares $s'\in\SqrsSetTagOf{k'}{\varrho{\text{\tiny{*}}}}$, $k'\geq1$, $k'\neq k$, for which $s \cap s' \neq \emptyset$.

For the case in which there is exactly one such square $s'$, we need an assumption on the choice of $\SqrsSetOf{\ell}{\varrho{\text{\tiny{*}}}}$ for all $\ell > 1$.
We assume that they are chosen such that if $s \in \SqrsSetTagOf{\ell}{\varrho{\text{\tiny{*}}}}$, $\ell \geq 1$, then $2s$, the square with edge of size $2^{\ell+1}$ sites, which is centred at the center of $s$, is in $\SqrsSetOf{\ell+1}{\varrho{\text{\tiny{*}}}}$. It is safe to assume it, since $s \in \SqrsSetTagOf{\ell}{\varrho{\text{\tiny{*}}}}$, and thus $\dist(s, (\varrho{\text{\tiny{*}}} - s\cap\varrho{\text{\tiny{*}}})) > 2^{\ell+1}$.
It is now straightforward to see that by Proposition~\ref{prop:a_s_construction}, the set of edges with non-zero gradients of $a_{s, \varrho{\text{\tiny{*}}}}$ and the set of edges with non-zero gradients of $a_{s', \varrho{\text{\tiny{*}}}}$ are disjoint.

For the case in which there are more than two distinct squares, let $s_i \in \SqrsSetTagOf{k_i}{\varrho{\text{\tiny{*}}}}$, $i=1,2$, be such distinct squares. We assume without loss of generality that $k_1\leq k_2 < k$.
Then by definition, $\dist(s_1, s_2\cap\varrho{\text{\tiny{*}}}) \geq 2M2^{\alpha (k_1+1)}$, but also $\dist(s_1, s_2\cap\varrho{\text{\tiny{*}}}) \leq 2^{k+1}$, and thus, by \eqref{eq:alphaMBounds}, $k > k_1 + 16$.
It is straightforward to see by Proposition~\ref{prop:a_s_construction} that the sets of edges with non-zero gradients are disjoint. This finishes the proof of Proposition \ref{prop:aExistance}.

\subsection{Complex Translation} \label{sec:complex}

We show in this section a proof to \eqref{eq:ComplexTrans} by a change of integration contour.
Let $\uptau, a:\Lambda\to\mathbb{R}$ with $\sum_{j\in\Lambda}\uptau_j=0$.
Note that $\sum_{j\in\Lambda}\uptau_j=0$ implies that \eqref{eq:ComplexTrans} depends only on gradients of $a$, and thus it is safe to assume that $a_v=0$, where $v$ is the normalization point of $\midaLambdaV$.

We start with the following one-variable equality.

\begin{lemma} \label{lemma:contour}
Let $a,q\in \mathbb{R}$, $\beta>0$, and let $(x_j)$ be a finite real sequence. Then
\begin{equation}
\int_{-\infty}^{\infty} e^{iqx} \exp[-(\beta/2)\sum_j (x-x_j)^2]dx = \int_{-\infty}^{\infty} e^{iq(x+ia)} \exp[-(\beta/2)\sum_j (x+ia-x_j)^2]dx.
\end{equation}
\end{lemma}
\begin{proof}
Let $b \in (0,\infty)$. As the function $I:\mathbb{C}\to\mathbb{C}$, $I(z):= e^{iqz}\exp[-(\beta/2)\sum_j (z-x_j)^2]$ is entire, its integral along the rectangular contour with corners at $-b, b, b+ia, -b+ia$ is zero. Therefore,
\begin{equation}\label{label:contour}
\int_{-b \rightarrow b} I(z)dz =
\int_{-b+ia \rightarrow b+ia} I(z)dz
+\int_{-b \rightarrow -b+ia} I(z)dz +
\int_{b+ia \rightarrow b} I(z)dz.
\end{equation}
Since
$$
\abs{I(z)} = e^{-q \text{Im}(z) - \frac{\beta}{2} \sum_j \text{Re}(z-x_j)^2},
$$
which tends to $0$ as $z\to\infty$ with $\text{Im}(z)$ bounded, it holds that that the last two terms on the right-hand side of \eqref{label:contour} tend to $0$ as $b\to\infty$, which finishes the proof.
\end{proof}

Iterative application of Lemma~\ref{lemma:contour} for each $\phi_k$, $k\in\Lambda\setminus\{v\}$ (recall that we assumed $a_v=0$), yields
\begin{multline*}
\int e^{i\phiOf{\uptau}}\midaLambdaV =
\int e^{i\langle \phi + ia, \uptau\rangle - \frac{\beta}{2}\langle \phi + ia, -\Delta_\Lambda (\phi+ia)\rangle}\mathds{1}_{[-\pi, \pi)}\big(\phi_v\big)\prod_{k \in \Lambda}d\phi_k \\
=
e^{-E_\beta(a,\uptau)}  \int e^{i\phiOf{\uptau+\beta\Delta_\Lambda a}} \midaLambdaV,
\end{multline*}
which finishes the proof of \eqref{eq:ComplexTrans}.

\section{Integer-Valued Discrete Gaussian Free Field}

In this section we deduce Theorem~\ref{thm:IV} from Theorem \ref{thm11} and prove Proposition~\ref{prop:IVUpperBound}.

\subsection{Proof of Theorem~\ref{thm:IV}} \label{sec:proofIV}

Let $\varepsilon > 0$ and $\Lambda$ be a square domain with free or periodic b.c.. Let $(F_N)$, $N\ge 1$, be the sequence of Fej\'er kernels, where $F_N:\mathbb{R}\to\mathbb{R}$ is given by
$$
F_N(x) := 1 + \sum_{q=1}^N 2\left( 1- \frac{q}{N}\right) \cos(qx).
$$
With a slight abuse of notation we set, for each $N\ge 1$, $\mathbb{E}_{\beta, \Lambda, F_N, v}:=\mathbb{E}_{\beta, \Lambda, \lambda_\Lambda, v}$ where $\lambda_j = F_N$ for all $j\in\Lambda$. Note that $F_N$ is $(1,0,0)$-sub-Gaussian for each $N\ge 1$. Thus, for each $N\ge 1$, we may apply Theorem \ref{thm11} to obtain, with $\beta_0$ uniform in $N$ and $\Lambda$, that for every $f:\Lambda\to\mathbb{R}$ with $\sum_{j\in\Lambda} f_j=0$,
\begin{equation*}
  \mathbb{E}_{\beta/(2\pi)^2, \Lambda, F_N, v}\left[ e^{\phiOf{\frac{1}{2\pi} f}} \right] \ge \exp\left[ \frac{1}{2(1+\varepsilon)\beta}\langle f, -\Delta_\Lambda^{-1}f\rangle \right],\quad \beta<(2\pi)^2\beta_0.
\end{equation*}
Thus, to prove the first part of Theorem~\ref{thm:IV} it suffices to show that
\begin{equation}\label{eq:IV_convergence}
  \mathbb{E}_{\beta/(2\pi)^2, \Lambda, F_N, v}\left[ e^{\phiOf{\frac{1}{2\pi} f}} \right]\xrightarrow{N\rightarrow\infty}
\mathbb{E}_{\beta, \Lambda, v}^{\IV}\left[ e^{\langle m, f\rangle} \right].
\end{equation}
Define for every $m\in \mathbb{Z}^{\Lambda}$ with $m_v = 0$,
$$
\Omega_m := \big\{ \phi\in \mathbb{R}^{\Lambda} : \phi_j \in [-\pi + 2\pi m_j, \pi + 2\pi m_j)\big\}.
$$
As $F_N$ is a summability kernel (see \cite[Chapter 1]{lambdaLim}) it follows that for every $g:\Lambda\to\mathbb{R}$ and $m\in \mathbb{Z}^{\Lambda}$ with $m_v=0$,
$$
\int_{\Omega_m}
e^{\phiOf{g}-\frac{\beta}{2}\langle\phi, (-\Delta_\Lambda)\phi\rangle}
\prod_{j\in\Lambda} F_N(\phi_j)d\phi_j
\xrightarrow{N\rightarrow\infty}
e^{\langle 2\pi m, g\rangle-\frac{\beta}{2}\langle2\pi m, (-\Delta_\Lambda)2\pi m\rangle}.
$$
Therefore, using dominated convergence,
\begin{align*}
\lim_{N\to\infty}\int e^{\phiOf{g}}\prod_{j \in \Lambda} F_N(\phi_j)\midaLambdaV
&=
\lim_{N\to\infty}\sum_{m\in \mathbb{Z}^{\Lambda},\,m_v = 0} \int_{\Omega_m}e^{\phiOf{g}-\frac{\beta}{2}\langle\phi, (-\Delta_\Lambda)\phi\rangle}
\prod_{j\in\Lambda} F_N(\phi_j)d\phi_j\\
&=
\sum_{m\in \mathbb{Z}^{\Lambda},\,m_v = 0}e^{\langle 2\pi m, g\rangle-\frac{\beta}{2}\langle2\pi m, (-\Delta_\Lambda)2\pi m\rangle},
\end{align*}
which implies \eqref{eq:IV_convergence}.

The second part of Theorem~\ref{thm:IV} now follows from the first as follows. Let $\delta>0$, let $g:\Lambda \rightarrow \mathbb{R}$ with $\sum_{j \in \Lambda} g_j = 0$ and set $f:=\delta\cdot g$. On the one hand we have from \eqref{eq:IVThm1} that
\begin{equation*}
  \mathbb{E}_{\beta, \Lambda, v}^{\IV}\big[ e^{\delta \mOf{g}} \big]\ge \exp\Big[ \frac{\delta^2}{2(1+\varepsilon)\beta}\langle g, -\Delta_\Lambda^{-1}g\rangle \Big]\ge 1 + \frac{\delta^2}{2(1+\varepsilon)\beta}\langle g, -\Delta_\Lambda^{-1}g\rangle.
\end{equation*}
On the other hand, a Taylor expansion in $\delta$ and the fact that $\mu_{\beta, \Lambda, v}^{\IV}$ is invariant under the mapping $m\mapsto-m$ show that
\begin{equation*}
  \mathbb{E}_{\beta, \Lambda, v}^{\IV}\big[ e^{\delta \mOf{g}} \big]\le \mathbb{E}_{\beta, \Lambda, v}^{\IV}\big[ 1 + \delta \mOf{g} + \frac{1}{2}\delta^2 \mOf{g}^2 + \frac{1}{6}\delta^3 |\mOf{g}|^3e^{\delta |\mOf{g}|}\big] \le 1+\frac{\delta^2}{2} \mathbb{E}_{\beta, \Lambda, v}^{\IV}\big[ \mOf{g}^2 \big] + D \delta^3,
\end{equation*}
for some $D=D(g)<\infty$. The last two inequalities imply the second part of Theorem~\ref{thm:IV}.


\subsection{Upper Bound on the Integer-Valued Discrete Gaussian Free Field}\label{sec:upper_bound}
In this section we prove Proposition~\ref{prop:IVUpperBound}.

Let $\Lambda$ be a finite, connected graph, $\beta>0$, $v\in\Lambda$ and $f:\Lambda\to\mathbb{R}$ satisfy $\sum_{j\in\Lambda} f_j=0$. For each $\eta\ge 0$ define the `Sine-Gordon' measure
$$
d\mu_{\beta, \Lambda, \eta, v}^{\UP}(\phi) := \frac{1}{Z_{\beta, \Lambda, \eta, v}^{\UP}} e^{\eta\sum_{j\in\Lambda}\cos\phi_j}\midaLambdaV,
$$
where the normalization constant $Z_{\beta, \Lambda, \eta, v}^{\UP}$ normalizes $\mu_{\beta, \Lambda, \eta, v}^{\UP}$ to be a probability measure.
Denote by $\mathbb{E}_{\beta, \Lambda, \eta, v}^{\UP}$ the corresponding expectation.
As
$$
k_{\eta_n}(t) := \frac{1}{2\pi} \frac{e^{\eta_n\cos t}}{\int_{-\pi}^{\pi}e^{\eta_n\cos t}dt}
$$
is a summability kernel for every sequence of non-negative real numbers $(\eta_n)_{n=1}^{\infty}$ increasing to infinity (see \cite[Chapter 1]{lambdaLim}) then, as in the proof of \eqref{eq:IV_convergence},
\begin{equation} \label{eq:UpMeasureLimit}
\lim_{\eta\to\infty} \mathbb{E}_{\beta/(2\pi)^2, \Lambda, \eta, v}^{\UP} \left[ e^{\phiOf{\frac{1}{2\pi}f}} \right] = \mathbb{E}_{\beta, \Lambda, v}^{\IV} \left[ e^{\langle m, f\rangle} \right].
\end{equation}
The key fact is that
\begin{equation}\label{eq:decreasing_function}
  \text{$\mathbb{E}_{\beta, \Lambda, \eta, v}^{\UP}\left[ e^{\phiOf{f}} \right]$ is a non-increasing function in $\eta\in[0,\infty)$}.
\end{equation}
Proposition~\ref{prop:IVUpperBound} is an immediate corollary of \eqref{eq:decreasing_function} and \eqref{eq:UpMeasureLimit}, by \eqref{eq:DGFFRes} and the fact that
$$
\mathbb{E}_{\beta, \Lambda, 0, v}^{\UP}\left[ e^{\phiOf{f}} \right] =
\mathbb{E}_{\beta, \Lambda, v}^{\GFF}\left[ e^{\langle f, \phi\rangle} \right].
$$
We proceed to prove \eqref{eq:decreasing_function}. It is straightforward to check that
$$
\frac{d}{d\eta} \mathbb{E}_{\beta, \Lambda, \eta, v}^{\UP}\left[ e^{\phiOf{f}} \right] =
\sum_{j\in\Lambda} \int \int e^{\phiOf{f}} \left( \cos\phi_j - \cos\phi'_j \right) d\mu_{\beta, \Lambda, \eta, v}^{\UP}(\phi) d\mu_{\beta, \Lambda, \eta, v}^{\UP}(\phi').
$$
Now, the change of variables
$$
\psi_j := \frac{\phi_j - \phi'_j}{\sqrt{2}}, \quad \chi_j := \frac{\phi_j + \phi'_j}{\sqrt{2}},
$$
yields, using the identities $\cos x - \cos y = -2 \sin \frac{x+y}{2} \sin\frac{x-y}{2}$, $\cos x + \cos y = 2 \cos \frac{x+y}{2} \cos\frac{x-y}{2}$, the fact that the Jacobian equals 1, the Taylor expansion of the exponential function and dominated convergence,
\begin{multline*}
\frac{d}{d\eta} \mathbb{E}_{\beta, \Lambda, \eta, v}^{\UP}\left[ e^{\phiOf{f}} \right]\\
=
-2 \left(\frac{1}{Z_{\beta, \Lambda, \eta, v}^{\UP}}\right)^2
\sum_{j\in\Lambda} \int \int
e^{\frac{1}{\sqrt{2}}(\psiOf{f} + \langle \chi, f\rangle)}
\sin\frac{\psi_j}{\sqrt{2}} \sin\frac{\chi_j}{\sqrt{2}}
e^{2\eta\sum_{\ell\in\Lambda}\cos\frac{\psi_\ell}{\sqrt{2}} \cos\frac{\chi_\ell}{\sqrt{2}}}
\midaLambdaVof{\psi} \midaLambdaVof{\chi}\\
=
-2
\sum_{j\in\Lambda}\sum_{n:\Lambda\to\{0,1,2,\ldots\}} \int \int
F_{j,n}(\psi)F_{j,n}(\chi)
\midaLambdaVof{\psi} \midaLambdaVof{\chi}\\
=
-2
\sum_{j\in\Lambda}\sum_{n:\Lambda\to\{0,1,2,\ldots\}}
\left(\int F_{j,n}(\psi)\midaLambdaVof{\psi}\right)^2 \le 0
,
\end{multline*}
where
\begin{equation*}
F_{j,n}(\psi) :=
\frac{1}{Z_{\beta, \Lambda, \eta, v}^{\UP}}
e^{\frac{1}{\sqrt{2}}\psiOf{f}}
\sin\frac{\psi_j}{\sqrt{2}}
\prod_{\ell\in\Lambda}\frac{\left(\sqrt{2\eta}\cos\frac{\psi_\ell}{\sqrt{2}}\right)^{n_\ell}}{\sqrt{n_\ell!}}. \qedhere
\end{equation*}

\section{Villain Model}

In this section we prove Theorem~\ref{thm:Vil}.
\subsection{Duality with the Integer-Valued Discrete Gaussian Free Field Model}
The starting point for our discussion is the well-known duality relation between plane rotator models and integer-valued height functions (see \cite[Appendix A]{fs} or \cite[Section 2.9]{ri}). For completeness we detail below the duality between the Villain model with zero boundary conditions and the integer-valued Gaussian free field with free boundary conditions, as used in the proof of Theorem~\ref{thm:Vil}.

For $L>1$, denote by $\Lambda_L^*$ the planar dual graph to $\Lambda_L^{\zero}$ (disregarding the outer face), which is defined to be the graph $\big(V(\Lambda_L^*), E(\Lambda_L^*)\big)$,
\begin{align*}
V(\Lambda_L^*) &:= \{ (i \pm 1/2, j \pm 1/2): (i,j) \in \Lambda_L^{\zero}\setminus\{\zv\}\},\\
E(\Lambda_L^*) &:= \{ \{(a,b), (c,d)\}: \{ \abs{c-a}, \abs{d-b}\} = \{0,1\} \}.
\end{align*}
Note that $\Lambda_L^*$ is isomorphic to $\Lambda_{L+1}^{\free}$. Note also the planar duality between the edges of $\Lambda_L^*$ and $\Lambda_{L+1}^{\free}$. For instance, $e = \{(a,b),(a+1,b)\}\in E(\Lambda_L^*)$ is dual to $e^* = \{(a+\frac{1}{2}, b-\frac{1}{2}),(a+\frac{1}{2}, b+\frac{1}{2})\}$ (which it crosses in the plane) and similarly for the other edges.

\begin{prop} \label{prop:VilIVRel}
Let $L>1$ and let $v\in\Lambda_L^*$. Then for each $x=(x_0, x_1)\in\Lambda_L^{\zero}\setminus\{\zv\}$ it holds that
$$
\mathbb{E}_{\beta, \Lambda_L^{\zero}}^{\Vil}\big[ \cos\theta_x \big] =
\mathbb{E}_{\beta^{-1}, \Lambda_L^*,v}^{\IV}\left[
\exp\left( -\beta^{-1}\sum_{t = x_0}^{L}\left(m_{(t+1/2,x_1-1/2)} - m_{(t+1/2,x_1+1/2)} + \frac{1}{2}\right) \right)
\right].
$$
\end{prop}

\begin{proof}

Fix an orientation $\vec{E}=\vec{E}(\Lambda_L^{\zero})$ for the edges $E(\Lambda_L^{\zero})$, in which all horizontal edges of the subgraph $\Lambda_L^{\free}$ are oriented to the left i.e., $((a,b),(c,d))\in\vec{E}$ if $a = c+1, b=d$, $(a,b),(c,d)\in\Lambda_L^{\free}$, edges $\{(L-1, a), z\}$, $(L-1, a)\in\Lambda_L^{\free}$, are oriented from $\zv$ to $(L-1, a)$, and all other edges are oriented arbitrarily. We also fix the orientation of the dual edges accordingly, for instance, the dual $e^*$ to $e = \{(a,b),(a+1,b)\}\in E(\Lambda_L^*)$, which is oriented to the left, is oriented up.

Let $x = (x_0, x_1)\in\Lambda_L^{\zero}$, and define $\chi=\chi_x:\vec{E}(\Lambda_L^{\zero}) \to \mathbb{R}$ by
$$
\chi_{(k,\ell)} :=
\threepartdef
{1}{\ell = (j, x_1), \ k=(j+1,x_1), \  x_0 \leq j, }
{1}{\ell = (L-1, x_1), \ k=\zv,}
{0}{\text{otherwise.}}
$$
Since $\mu_{\beta, \Lambda_L^{\zero}}^{\Vil}$ is invariant under the mapping $\theta\mapsto-\theta$, it holds that
\begin{multline} \label{eqVil1}
\mathbb{E}_{\beta, \Lambda_L^{\zero}}^{\Vil}\big[ \cos\theta_x \big] =
\mathbb{E}_{\beta, \Lambda_L^{\zero}}^{\Vil}\Big[ e^{i(\theta_x - \theta_{\zv})} \Big] \\
= \frac{1}{Z_{\beta, \Lambda_L^{\zero}}^{\Vil}}\int
\prod_{(k,\ell)\in\vec{E}} \sum_{m \in \mathbb{Z}} e^{-\frac{\beta}{2}(\theta_k - \theta_\ell + 2\pi m)^2 - i\chi_{(k,\ell)}(\theta_k - \theta_\ell)} \delta_{0}(\theta_{\zv})
\prod_{j \in \Lambda_L^{\zero}\setminus\{\zv\}} \mathds{1}_{[-\pi, \pi)}\big(\theta_j\big) d\theta_j.
\end{multline}
Substituting in \eqref{eqVil1}
$$
f_\chi(\theta) := \sum_{m\in\mathbb{Z}} e^{-\frac{\beta}{2}(\theta+2\pi m)^2 - i\chi\theta} = \frac{1}{\sqrt{2\pi\beta}} \sum_{n\in\mathbb{Z}} e^{-\frac{1}{2\beta} (n+\chi)^2}e^{i n \theta},
$$
which holds by application of Fourier transform,
$$
\hat{f}_\chi(n) = \frac{1}{2\pi} \sum_{m\in\mathbb{Z}} \int_{-\pi}^\pi e^{-\frac{\beta}{2}(\theta+2\pi m)^2 - i\chi\theta -in\theta} d\theta =
\frac{1}{2\pi} \int_{-\infty}^\infty e^{-\frac{\beta}{2}\theta^2 - i\chi\theta -in\theta} d\theta = \frac{1}{\sqrt{2\pi\beta}} e^{-\frac{1}{2\beta} (n+\chi)^2}.
$$
yields
\begin{equation*}
\begin{split}
&\mathbb{E}_{\beta, \Lambda_L^{\zero}}^{\Vil}\big[ \cos\theta_x\big] =
\frac{1}{\sqrt{2\pi\beta}Z_{\beta, \Lambda_L^{\zero}}^{\Vil}}
\int \prod_{(k,\ell)\in\vec{E}}
\sum_{n \in \mathbb{Z}}
e^{-\frac{1}{2\beta} (n+\chi_{(k,\ell)})^2}
e^{i n (\theta_k-\theta_\ell)}
\delta_0(\theta_{\zv})\prod_{j \in \Lambda_L^{\zero}\setminus\{\zv\}} \mathds{1}_{[-\pi, \pi)}\big(\theta_j\big) d\theta_j \\
&= \frac{1}{\sqrt{2\pi\beta}Z_{\beta, \Lambda_L^{\zero}}^{\Vil}}
\sum_{n:\vec{E}\to\mathbb{Z}}
\prod_{(k,\ell)\in\vec{E}}
e^{-\frac{1}{2\beta} (n_{(k,\ell)}+\chi_{(k,\ell)})^2}
\int e^{i \sum_{(k,\ell)\in\vec{E}}n_{(k,\ell)} (\theta_k-\theta_\ell)} \delta_0(\theta_{\zv})\prod_{j \in \Lambda_L^{\zero}\setminus\{\zv\}} \mathds{1}_{[-\pi, \pi)}\big(\theta_j\big) d\theta_j \\
&= \frac{1}{\sqrt{2\pi\beta}Z_{\beta, \Lambda_L^{\zero}}^{\Vil}}
\sum_{n:\vec{E}\to\mathbb{Z}}
\prod_{(k,\ell)\in\vec{E}}
e^{-\frac{1}{2\beta} (n_{(k,\ell)}+\chi_{(k,\ell)})^2}
\prod_{j \in \Lambda_L^{\zero}\setminus\{\zv\}} \int e^{i (\delta n)_j \theta_j} \mathds{1}_{[-\pi, \pi)}\big(\theta_j\big)d\theta_j\\
&= \frac{1}{\sqrt{2\pi\beta}Z_{\beta, \Lambda_L^{\zero}}^{\Vil}}
\sum_{n:\vec{E}\to\mathbb{Z}}
\prod_{(k,\ell)\in\vec{E}}
e^{-\frac{1}{2\beta} (n_{(k,\ell)}+\chi_{(k,\ell)})^2}
\mathds{1}_{\delta n = \mathbf{0}},
\end{split}
\end{equation*}
where we set $n_{(k,\ell)} := -n_{(\ell,k)}$ whenever $(\ell,k)\in\vec{E}$, and $\delta n:\Lambda_L^{\zero} \to \mathbb{Z}$ is defined by
$$
(\delta n)_k := \sum_{(k,\ell)\in\vec{E}} n_{(k,j)}.
$$

We claim that there is a 1-1 correspondence between $\{n:\vec{E}\to \mathbb{Z}, \delta n = 0\}$, and integer-valued variables $\{ m_p \}_{p\in\Lambda_L^*}$, where $m_{v}$ is set to 0.
It is defined by $n_{(k,\ell)} = m_{k'} - m_{\ell'}$, where $(k',\ell')$ is the corresponding oriented edge.
Indeed, given $\{ m_p \}$, we set $n_{k,\ell} = m_{k'} - m_{\ell'}$ for every $(k,\ell)$, and it is straightforward to see that $(\delta n) =\mathbf{0}$.
Conversely, given $n$ with $(\delta n) =\mathbf{0}$, we set $m_p = \sum_{i=0}^t n_{(p_i, p_{i+1})}$, where $v=p_0,p_1,\ldots, p_t,p_{t+1}=p$, is a path from $v$ to $p$. Since $\delta n = 0$ it is straightforward to see that $m_p$ is uniquely defined.

Therefore:
\begin{align*}
&\mathbb{E}_{\beta, \Lambda_L}^{\Vil}\big[ \cos \theta_x \big]
=
\frac{1}{\sqrt{2\pi\beta}Z_{\beta, \Lambda_L^{\zero}}^{\Vil}}
\sum_{\substack{m:\Lambda_L^*\rightarrow\mathbb{Z} \\ m(v)=0}}
\prod_{(k,\ell)\in\vec{E}}
e^{-\frac{1}{2\beta} (m_{k'} - m_{\ell'}+\chi_{(k,\ell)})^2} \\
&=
\frac{1}{\sqrt{2\pi\beta}Z_{\beta, \Lambda_L^{\zero}}^{\Vil}}
\sum_{\substack{m:\Lambda_L^*\rightarrow\mathbb{Z} \\ m(v)=0}}
e^{-\frac{1}{\beta}\sum_{t = x_0}^{L}\left(m_{(t+1/2,x_1-1/2)} - m_{(t+1/2,x_1+1/2)} + \frac{1}{2}\right)}
e^{-\frac{1}{2\beta} \sum_{(k,\ell)\in\vec{E}(\Lambda_L^*)}(m_{k'} - m_{\ell'})^2},
\end{align*}
We finish the proof by observing that a similar calculation yields
\begin{equation*}
Z_{\beta, \Lambda_L^{\zero}}^{\Vil} =
\frac{1}{\sqrt{2\pi\beta}Z_{\beta, \Lambda_L^{\zero}}^{\Vil}}
\sum_{\substack{m:\Lambda_L^*\rightarrow\mathbb{Z} \\ m(v)=0}}
e^{-\frac{1}{2\beta} \sum_{(k',\ell')\in\vec{E}(\Lambda_L^*)}(m_{k'} - m_{\ell'})^2}. \qedhere
\end{equation*}

\end{proof}

\subsection{Proof of Theorem~\ref{thm:Vil}}
To prove Theorem~\ref{thm:Vil} we first prove the following theorem.
\begin{thm} \label{thm54}
For any $\Gamma > 0$, $\eta\in\mathbb{R}$, $0 \leq \theta < 1/16$, there exists $\beta_0 = \beta_0(\Gamma,\eta, \theta) > 0$ such that the following holds.
Let $L>1$, let $v\in\Lambda_L^{\free}$, let $(\lambda_j)$, $j\in\Lambda$, be a collection of $\SG$ functions, let $\beta < \beta_0$ and let $y=(y_0, y_1)\in\Lambda_L^{\free}$ with $y_1<L-1$.
Then there exists a positive absolute constant $c>0$ such that
\begin{equation}\label{eq:thm54}
\mathbb{E}_{\beta, \Lambda_L^{\free}, \lambda_\Lambda, v}
\left[
\exp\Big[ -2\pi\beta\sum_{t= y_0}^{L}\Big(\phi_{(t,y_1)} - \phi_{(t,y_1+1)}+\frac{2\pi}{2}\Big) \Big]
\right] \geq (L-y_0+1)^{-c\beta}.
\end{equation}
\end{thm}

Theorem~\ref{thm:Vil} immediately follows from Theorem~\ref{thm54}. Indeed, let $L > 1$ and $y=(y_0,y_1), v\in\Lambda_L^{\free}$ with $y<L-1$. As in the proof of Theorem~\ref{thm:IV} (see~\eqref{eq:IV_convergence}), for any $\beta>0$,
\begin{equation*}
\mathbb{E}_{\beta/(2\pi)^2, \Lambda_L^{\free},F_N, v}\left[e^{-\frac{\beta}{2\pi}\sum_{t=y_0}^{L}\phi_{(t,y_1)} - \phi_{(t,y_1+1)}}
\right] \xrightarrow{N\to\infty}
\mathbb{E}_{\beta, \Lambda_L^{\free}, v}^{\IV}\left[
e^{-\beta\sum_{t = y_0}^{L}m_{(t,y_1)} - m_{(t,y_1+1)}}\right].
\end{equation*}
Therefore, applying Theorem~\ref{thm54} with $\beta$ replaced by $\frac{\beta}{(2\pi)^2}$ we conclude that
\begin{equation*}
\mathbb{E}_{\beta, \Lambda_L^{\free}, v}^{\IV}\left[
\exp\left( -\beta\sum_{t = y_0}^{L}\left(m_{(t,y_1)} - m_{(t,y_1+1)} + \frac{1}{2}\right) \right)\right]\ge (L-y_0+1)^{-\frac{c\beta}{(2\pi)^2}},\quad 0 < \beta<\beta_0 \cdot (2\pi)^2.
\end{equation*}
Comparing the last inequality with Proposition~\eqref{prop:VilIVRel}, and recalling that $\Lambda_L^*$ is isomorphic to $\Lambda_{L+1}^{\free}$, implies Theorem~\ref{thm:Vil} for vertices $x = (x_0, x_1)\in\Lambda_L^{\zero}$ with $\dist(x,z) = L-x_0$. The statement of Theorem~\ref{thm:Vil} for general $x\in\Lambda_L^{\zero}$ follows from the lattice symmetries.

Therefore we are left with proving Theorem~\ref{thm54}. To this end, fix $\Gamma > 0$, $\eta\in\mathbb{R}$, $0 \leq \theta < 1/16$, $L>1$, $v\in\Lambda_L^{\free}$, a collection of $\SG$ functions $(\lambda_j)$, ${j\in\Lambda}$, and $y=(y_0, y_1)\in\Lambda_L^{\free}$ with $y_1<L-1$. For notational simplicity, denote by $\Delta_L^{\f}$ the graph Laplacian of $\Lambda_L^{\free}$.

The proof is very similar to the proof of Theorem~\ref{thm11} (see Section~\ref{sec:proof11}). To obtain an equality similar to \eqref{eq:varChangeInE}, we define $f_y:\Lambda_L^{\free}\to\mathbb{R}$ by
\begin{equation}\label{eq:f_y_def}
f_y(x_0,x_1) = \threepartdef
{1}{y_0 \leq x_0 < L,\ x_1=y_1,}
{-1}{y_0 \leq x_0 < L,\ x_1=y_1+1,}
{0}{\text{otherwise}.}
\end{equation}
Note that the observable in the left-hand side of \eqref{eq:thm54} is now of the form
$$
\exp\Big[ -2\pi\beta\langle \phi, f_y \rangle -\frac{(2\pi)^2 \beta(L-y_0)}{2} \Big].
$$
As in \eqref{eq:sigmaDef}, define $\sigma^{\Vil}: \Lambda_L^{\free} \rightarrow \mathbb{R}$ to be the solution of
$$
\twopartdef
{\Delta_L^{\f} \sigma^{\Vil} = 2\pi f_y,}{}
{\sigma_v^{\Vil} = 0.}{}
$$
Now, a change of variables $\phi_j \to \phi_j + \sigma_j^{\Vil}$ yields
\begin{equation}\label{eq:proofVil0}
\mathbb{E}_{\beta, \Lambda_L^{\free}, \lambda_\Lambda, v}\big[ e^{-2\pi\beta\langle \phi, f_y \rangle -\frac{(2\pi)^2 \beta(L-y_0)}{2}}\big] =
Z_{\beta, \Lambda_L^{\free}, \lambda_\Lambda, v}^{-1}
e^{\frac{(2\pi)^2\beta}{2} \langle f_y, (-\Delta_L^{\f})^{-1} f_y\rangle- \frac{(2\pi)^2\beta}{2}(L-y_0)}Z_{\beta, \Lambda_L^{\free}, \lambda_\Lambda, v}(\sigma^{\Vil}),
\end{equation}
where $Z_{\beta, \Lambda_L^{\free}, \lambda_\Lambda, v}(\sigma)$ is defined in \eqref{eq:ZOfSigmaDef}.
Let $\beta_0$ be the positive number given by Theorem~\ref{thm234}. The equality in that theorem states that
\begin{equation}\label{eq:proofVil1}
Z_{\beta, \Lambda_L^{\free}, \lambda_\Lambda, v}(\sigma^{\Vil}) = \sum_{\ens \in \mathscr{F}} c_\ens Z_\ens (\sigma^{\Vil}), \quad \beta < \beta_0,
\end{equation}
where $Z_\ens (\sigma)$ is defined in \eqref{eq:z_ens_def}.

Denote by $[a]_{2\pi}$ the unique $b\in[-\pi, \pi)$, such that $a-b$ is an integer multiple of $2\pi$. Same calculation to the calculation in Section~\ref{sec:proof11}, with $y$ is replaced with $[y]_{2\pi}$ in Claim~\ref{lemma:doubleCos}, yields
\begin{equation}\label{eq:proofVil2}
\frac{Z_\ens(\sigma^{\Vil})}{Z_\ens (0)}
\geq \exp \big(-\DFr\sum_{\varrho \in \ens} \abs{\zed}[\sigmaOfVil{\varrho}]_{2\pi}^2 \big).
\end{equation}

We require the following replacement for Claim~\ref{newlemma55eq} tailored to the structure of $\sigma^{\Vil}$.
\begin{claim} \label{newlemma55eqVil}
Let $D>0$. There exists $0 < \beta_1 \leq \beta_0$ such that
\begin{equation}
\sum_{\varrho \in \ens} \abs{\zed} \cdot [\langle \sigma^{\Vil}, \varrho \rangle]_{2\pi}^2 \leq
\frac{\beta}{D} \cdot \ln(L-y_0+1) ,\quad\beta<\beta_1, \quad \ens\in\mathscr{F}.
\end{equation}
\end{claim}

To finish the proof, we use the following claim.
\begin{claim} \label{claim:GreenId1}
There exists a positive absolute constant $\DSix$ such that
\begin{equation} \label{eqVilProof2}
\langle f_y, (-\Delta_L^{\f})^{-1}f_y\rangle \geq
(L-y_0) - \DSix\cdot\ln(L-y_0+1)\quad y=(y_0,y_1)\in\Lambda_L^{\free}, y_1 < L-1.
\end{equation}
\end{claim}

Applying Claim~\ref{newlemma55eqVil} with $D=\DFr$ and substituting back in \eqref{eq:proofVil0}, and using Claim~\ref{claim:GreenId1} yields
$$
\mathbb{E}_{\beta, \Lambda_L^{\free}, \lambda_\Lambda, v}\big[ D_{\Lambda_L^{\free}, y}^\beta(\phi) \big] \geq
\exp \Big[ -\beta\ln(L-y_0+1)(1+(2\pi)^2\DSix/2)\Big],
$$
which finishes the proof of Theorem~\ref{thm54}.

\subsection{Proof of Claims \ref{newlemma55eqVil} and \ref{claim:GreenId1}}
The proof of the two claims relies on the properties of the Green function $(-\Delta_L^{\f})^{-1}$. It is technically convenient to replace this Green function by $(-\Delta_{2L}^{\p})^{-1}$, where $\Delta_{2L}^{\p}$ defined as the graph Laplacian of $\Lambda_{2L}^{\per}$, as this latter Green function is invariant to translations. Our first task is thus to establish a link between the two Green functions.

Define a linear expansion operator $T_{\fp}:\mathbb{R}^{\Lambda_L^{\free}} \rightarrow \mathbb{R}^{\Lambda_{2L}^{\per}}$ by
$$
T_{\fp}[f]_{(a,b)} = f_{(\min\{a,2L-1-a\}, \min\{b,2L-1-b\})},\quad f:\Lambda_L^{\free}\to\mathbb{R}, (a,b)\in \Lambda_{2L}^{\per},
$$
so that $T_{\fp}[f]$ is formed by `reflecting $f$ across the principal axes'. The definition of $T_{\fp}$ is chosen to complement the definition \eqref{eq:Laplacian_def} of the Laplacian so that
\begin{equation*}
  \Delta_{2L}^{\p}\circ T_{\fp} = T_{\fp}\circ \Delta_L^{\f}.
\end{equation*}
This readily implies, taking care that $\Delta_\Lambda^{-1}$ is only a pseudoinverse as in \eqref{eq:pseudo-inverse_def}, that
\begin{equation}\label{eq:T_fp_eq}
  T_{\fp}\circ (\Delta_L^{\f})^{-1} = (\Delta_{2L}^{\p})^{-1}\circ T_{\fp}.
\end{equation}

With the above link in hand, we proceed to point out several simple properties of the Green function $(-\Delta_{2L}^{\p})^{-1}$. Our starting point is the equality
\begin{equation}\label{eq:Delta_with_finite_differences}
  \partial_1  \partial_1^T + \partial_2 \partial_2^T = -\Delta_{2L}^{\p},
\end{equation}
where $\partial_1, \partial_2:\mathbb{R}^{\Lambda_{2L}^{\per}}\to\mathbb{R}^{\Lambda_{2L}^{\per}}$ are the backward difference (linear) operators in the first and second coordinate, respectively, as given by
\begin{align*}
(\partial_1 f)_{(a,b)} &:= f(a,b)-f((a-1) \text{ mod } 2L, b),\quad f:\Lambda_L^{\free}\to\mathbb{R}, (a,b)\in \Lambda_{2L}^{\per},\\
(\partial_2 f)_{(a,b)} &:= f(a,b)-f(a, (b-1) \text{ mod } 2L),\quad f:\Lambda_L^{\free}\to\mathbb{R}, (a,b)\in \Lambda_{2L}^{\per}.
\end{align*}

With this definition, \eqref{eq:Delta_with_finite_differences} and using the translation invariance of $\Delta_{2L}^{\p}$, it follows in a straightforward manner that
\begin{gather}
\label{eq:PerId1}
(\partial_1)^T (-\Delta_{2L}^{\p})^{-1} \partial_1 + (\partial_2)^T (-\Delta_{2L}^{\p})^{-1} \partial_2 = I, \quad \text{ $I$ is the identity on $\{ f: \sum_{j\in\Lambda_{2L}} f(j) = 0\}$}, \\
\label{eq:PerId2}
\partial_1(\Delta_{2L}^{\p})^{-1}\partial_2 = \partial_2(\Delta_{2L}^{\p})^{-1}\partial_1.
\end{gather}

It also immediately follows by translation invariance that
$$
(\Delta_{2L}^{\p})^{-1}_{(0,0), (t,s)} = (\Delta_{2L}^{\p})^{-1}_{(a,b), (a+t \text{ mod } 2L, b+s \text{ mod } 2L)}, \quad (a,b), (t,s) \in \Lambda_{2L}^{\per},
$$
which allows us to reduce the Green function to the function $G_{2L}: \Lambda_{2L}^{\per}\to\mathbb{R}$ defined by
$$
G_{2L}(t,s) := (-\Delta_{2L}^{\p})^{-1}_{(0,0), (t,s)}, \quad (t,s)\in\Lambda_{2L}^{\per}.
$$
It is a well known fact that there exists a positive absolute constant $\DSev$ such that
\begin{equation} \label{eq:GreenLnBound}
G_{2L}(0,0)-G_{2L}(2a,0) \leq \DSev \ln(a+1),\quad a\geq0.
\end{equation}

To make use of \eqref{eq:T_fp_eq} we denote by $F_y:\Lambda_{2L}^{\p}\rightarrow\mathbb{R}$ the function
$$
F_y(a,b):=\threepartdef
{1}{y_0\leq a \leq 2L-1-y_0,\quad b=y_1,}
{-1}{y_0\leq a \leq 2L-1-y_0,\quad b=2L-2-y_1,}
{0}{\text{otherwise},}
$$
and note that
\begin{equation}\label{eq:F_y}
\partial_2F_y = T_{\fp}f_y.
\end{equation}

Recalling that $V(\Lambda_L^{\free}) = \{0,1,\ldots,L-1\}^2$ and $V(\Lambda_{2L}^{\per}) = \{0,1,\ldots,2L-1\}^2$, it is convenient to regard $\Lambda_L^{\free}$ as a subgraph of $\Lambda_{2L}^{\per}$.
We are now ready to prove Claim~\ref{claim:GreenId1}. We will also show that
\begin{equation} \label{eq:GreenId2}
\langle \partial_1 F_y, (-\Delta_{2L}^{\p})^{-1} \partial_1 F_y \rangle \leq
\DSix\cdot\ln(L-y_0+1).
\end{equation}

\begin{proof}[Proof of Claim~\ref{claim:GreenId1} and \eqref{eq:GreenId2}]

Denote by $\mathds{1}_{\Lambda_L^{\free}}:\Lambda_{2L}^{\per}\to\mathbb{R}$ the function satisfying $\mathds{1}_{\Lambda_L^{\free}}(j)=1$ for $j\in\Lambda_L^{\free}$ and $0$ otherwise. It is straightforward to see by \eqref{eq:F_y} that $\partial_2(F_y\cdot\mathds{1}_{\Lambda_L^{\free}})(j)=f_y(j)$ if $j\in\Lambda_L^{\free}$ and $\partial_2(F_y\cdot\mathds{1}_{\Lambda_L^{\free}})(j)=0$ otherwise. Therefore, by \eqref{eq:T_fp_eq} and \eqref{eq:F_y},
\begin{multline*}
\langle f_y, (-\Delta_L^{\f})^{-1}f_y\rangle =
\langle \partial_2 (F_y\cdot\mathds{1}_{\Lambda_L^{\free}}), (-\Delta_{2L}^{\p})^{-1}\partial_2 F_y \rangle =
\langle F_y\cdot\mathds{1}_{\Lambda_L^{\free}}, \partial_2^T(-\Delta_{2L}^{\p})^{-1}\partial_2 F_y \rangle \\
= \langle F_y\cdot\mathds{1}_{\Lambda_L^{\free}}, (I - \partial_1^T(-\Delta_{2L}^{\p})^{-1}\partial_1) F_y \rangle =
(L - y_0) - \langle \partial_1 (F_y\cdot\mathds{1}_{\Lambda_L^{\free}}), (-\Delta_{2L}^{\p})^{-1}\partial_1 F_y \rangle.
\end{multline*}
It is straightforward to check that (using \eqref{eq:GreenLnBound})
\begin{align*}
\langle \partial_1 (F_y\cdot\mathds{1}_{\Lambda_L^{\free}}), (-\Delta_{2L}^{\p})^{-1}\partial_1 F_y \rangle &= \gamma_{1}(y) + \gamma_{2}(y) \leq \DSev\cdot\ln(L-y_0 + 1) + \DSev\ln2,\\
\langle \partial_1 F_y, (-\Delta_{2L}^p)^{-1} \partial_1 F_y \rangle &= 4\gamma_{G,1}(y)\leq\DSev\cdot\ln(L-y_0 + 1),
\end{align*}
where
\begin{align*}
\gamma_{1}(y) &:= G_{2L}\big(0,0\big) - G_{2L}\big(0, 2(L-y_1)\big) - G_{2L}\big(2(L-1-y_0), 0\big) + G_{2L}\big(2(L-1-y_0), 2(L-y_1)\big),\\
\gamma_{2}(y) &:= G_{2L}\big(0,L-1-y_1\big) - G_{2L}\big(0, L+1-y_1\big) - G_{2L}\big(2(L-1-y_0), L-1-y_1\big) \\
&+ G_{2L}\big(2(L-1-y_0), L+1-y_1\big),
\end{align*}
which finishes the proof.\qedhere
\end{proof}

\begin{proof}[Proof of Claim~\ref{newlemma55eqVil}]
For notational simplicity, denote $\sigma := \sigma^{\Vil}$.
The proof is similar to the proof of Claim~\ref{newlemma55eq}. Our starting point is
\begin{equation*}
[\sigmaOf{\varrho}]_{2\pi}^2 \leq 2\abs{D(\varrho)} \sum_{\substack{j,\ell\in D(\varrho)\\j\thicksim\ell}} c_{\{j,\ell\}}^2[\sigma_j - \sigma_\ell]_{2\pi}^2,
\end{equation*}
for integer $c_{\{j,\ell\}}$ satisfying $|c_{\{j,\ell\}}|\leq\frac{1}{2}\norm{\varrho}_2^2$ (where we used $\abs{[a+b]_{2\pi}} \leq \abs{[a]_{2\pi}} + \abs{[b]_{2\pi}}$, $a,b\in\mathbb{R}$). The proof continues identically to the proof of Claim~\ref{newlemma55eq} (with $D\cdot \DSix$ replacing $D$) and yields
\begin{equation}\label{eq:claim63_eq}
\sum_{\varrho \in \ens}\abs{\zed} \cdot [\sigmaOf{\varrho}]_{2\pi}^2 \leq
\frac{\beta}{D\cdot \DSix}\sum_{j\thicksim\ell} [\sigma_j - \sigma_\ell]_{2\pi}^2.
\end{equation}

We continue by separately analysing $[\sigma_j - \sigma_\ell]_{2\pi}^2$ for horizontal and vertical edges $\{j,\ell\}$.
Suppose $\{j,\ell\}$ is horizontal and that $j$ has the larger first coordinate (i.e. $j_0 = \ell_0 + 1$). Then, by \eqref{eq:PerId2},
\begin{align*}
\sigma_j - \sigma_\ell &=
2\pi\big[(\Delta_L^{\f})^{-1}f_y(j) - (\Delta_L^{\f})^{-1}f_y(\ell)\big] \\
&= 2\pi\big[(\Delta_{2L}^{\p})^{-1}T_{\fp}f_y(j) - (\Delta_{2L}^{\p})^{-1}T_{\fp}f_y(\ell)\big]
= 2\pi\big(\partial_1(\Delta_{2L}^{\p})^{-1}T_{\fp}f_y\big)(j) \\
&= 2\pi\big(\partial_1(\Delta_{2L}^{\p})^{-1}\partial_2 F_y\big)(j)
= 2\pi\big(\partial_2(\Delta_{2L}^{\p})^{-1}\partial_1 F_y\big)(j)\\
&= 2\pi\big[\big((\Delta_{2L}^{\p})^{-1}\partial_1 F_y\big)(j)-(\Delta_{2L}^{\p})^{-1}\partial_1 F_y\big)(\ell)\big]
\end{align*}
Now suppose $\{j,\ell\}$ is vertical and that $j$ has the larger second coordinate (i.e. $j_1 = \ell_1 + 1$). Then, similarly to the previous case, using \eqref{eq:PerId1} and the fact that $F_y(j)$ is integer for $j\in\Lambda_{2L}^{\per}$,
\begin{align*}
[\sigma_j - \sigma_\ell]_{2\pi}^2 &=
[2\pi\big(\partial_2^T(\Delta_{2L}^{\p})^{-1}\partial_2 F_y\big)(j)]_{2\pi}^2 =
[2\pi\big(F_y - \partial_1^T(\Delta_{2L}^{\p})^{-1}\partial_1 F_y\big)(j)]_{2\pi}^2\\
&=[2\pi\big(\partial_1^T(\Delta_{2L}^{\p})^{-1}\partial_1 F_y\big)(j)]_{2\pi}^2 =
2\pi\big[\big((\Delta_{2L}^{\p})^{-1}\partial_1 F_y\big)(j)-(\Delta_{2L}^{\p})^{-1}\partial_1 F_y\big)(\ell)\big]_{2\pi}^2.
\end{align*}

Now, using \eqref{eq:GreenId2},
$$
\sum_{j\thicksim\ell} [\sigma_j - \sigma_\ell]_{2\pi}^2 \leq \gradNorm{\sigma} = \langle \partial_1 F_y, (-\Delta_{2L}^{\p})^{-1} \partial_1 F_y \rangle \leq \DSix \ln(L-y_0+1),
$$
which completes the proof by substituting the last inequality in \eqref{eq:claim63_eq}.
\end{proof}


\section{Discussion and Open Questions} \label{sec:openQ}
In this paper we presented the Fr\"ohlich-Spencer proof for delocalization of the two-dimensional integer-valued Gaussian free field at high temperature and for power-law decay of correlations in the two-dimensional plane rotator model with Villain interaction at low temperature. Our understanding of these and related models is still rather incomplete and in this section we mention some of the outstanding open problems (see also the discussion by Velenik \cite{V06}).

{\bf Scaling limit.} What is the scaling limit of the two-dimensional integer-valued discrete Gaussian free field with small $\beta$? The results of Fr\"ohlich and Spencer suggest that it converges to the (massless) continuum Gaussian free field but no proof is available.

{\bf Maximum.} Recall that $\Lambda_L^{\free}$ is the subgraph of $\mathbb{Z}^2$ with vertex set $\{0,1,\ldots, L-1\}^2$. The maximum of the two-dimensional \emph{real-valued} discrete Gaussian free field on $\Lambda_L^{\free}$ (defined in \eqref{label:mida}) is of order $\log L$ (with much more refined information available). Is the same true for the \emph{integer-valued} discrete Gaussian free field with small $\beta$? Proposition~\ref{prop:IVUpperBound} implies an upper bound with this order of magnitude.

{\bf Gibbs states of the plane rotator model.} It is known that the plane rotator model has a unique translation-invariant Gibbs state \cite{BFL77} at every temperature. It is expected that, in fact, it has a unique Gibbs state but this remains unproven.

{\bf Integer-valued random surfaces with general interaction potentials.} Let $u = (u_n)_{n\in\mathbb{Z}}$ be a sequence satisfying $u_n = u_{-n}$. One may define an integer-valued random surface with interaction function $u$ as follows. Given a finite, connected graph $\Lambda$ and vertex $v\in\Lambda$, the measure $\mu_{u, \Lambda,v}^{\IV}$ on functions $m:\Lambda\to\mathbb{Z}$ is given by
$$
d\mu_{u, \Lambda, v}^{\IV} =
\frac{1}{Z_{u, \Lambda,v}^{\IV}}
\exp\Bigg[-\sum_{j\thicksim\ell}u_{m_j - m_\ell}\Bigg] \cdot
d\delta_0(m_v)
\prod_{j \in \Lambda\setminus\{v\}} d_{\text{count}(\mathbb{Z})}(m_j)
$$
where $\delta_0$ is the Dirac delta measure at $0$, $d_{\text{count}(\mathbb{Z})}$ is the counting measure on $\mathbb{Z}$, and the normalization constant $Z_{u, \Lambda,v}^{\IV}$ normalizes the measure $\mu_{u, \Lambda, v}^{\IV}$ to be a probability measure, which is possible assuming suitable decay properties of $u$. The corresponding expectation is denoted by $\mathbb{E}_{u, \Lambda, v}^{\IV}$. An outstanding challenge is to analyze the fluctuations of such surfaces, e.g., when $\Lambda$ is a subset of the hypercubic lattice $\mathbb{Z}^d$. Versions of the standard Peierls argument apply when the ratio $u_{n}/ u_0$ grows sufficiently rapidly, and show that the surface is in a localized phase, with uniformly bounded variance at every vertex. In fact, localization is predicted to always (perhaps with mild growth assumptions on $u$) occur in dimensions $d\ge 3$ but is generally unknown. The main result in this direction is provided by G\"opfert and Mack \cite{GM82} who proved that the the integer-valued discrete Gaussian free field in dimension $d=3$ is localized and, remarkably, has exponential decay of correlations (positive mass) at all positive inverse temperatures $\beta$. In contrast, one expects the random surface to be delocalized in two dimensions for a large class of interaction sequences $u$, but the only available results are those of Fr\"ohlich and Spencer \cite{fs} on the integer-valued discrete Gaussian free field and the Solid-On-Solid model (when $u_n = |n|$). It is of interest to extend the methods of \cite{fs} to more general interaction potentials (some steps in this direction are in \cite[Lemma 4.3, Section 6 and Section 7]{fs}). In this regard we put forward the following conjecture.
\begin{conj}
  Let $U:\mathbb{R}\to\mathbb{R}$ be smooth and satisfy $U(x)=U(-x)$ and $\lim_{x\to\infty} U(x)/x^{\alpha}=\infty$ for some $\alpha>0$. For each $M>0$ define $u^M = (u^M_n)$ by $u^M_n := U(n / M)$. Then there exists some $M_0(U)>0$ such that for each $M>M_0(U)$,
  \begin{equation*}
    \lim_{L\to\infty} \frac{\mathbb{E}_{u^M, \Lambda_L^{\free}, (0,0)}^{\IV} \left(m_{(L-1,L-1)}^2\right)}{\log(L)} > 0.
  \end{equation*}
\end{conj}
In the special case that $U(x) = |x|^\alpha$ for some $\alpha>0$ the conjecture states that the models with $u_n = \beta |n|^\alpha$ will delocalize in two dimensions at sufficiently low $\beta$ (depending on $\alpha$), extending the results of Fr\"ohlich and Spencer which apply to $\alpha=1,2$.

\medskip
{\bf Acknowledgements.} We thank J\"urg Fr\"ohlich for explaining to the second author the proof of Proposition~\ref{prop:IVUpperBound} from his work with Park \cite{upperBound}.

\end{document}